\documentclass[11pt]{article}

\usepackage{amsmath}
\usepackage{amssymb}
\usepackage{amsthm}
\usepackage{mathabx}
\usepackage{graphicx} 
\usepackage{fullpage}
\usepackage{hyperref}
\usepackage{thmtools}
\usepackage[nameinlink,capitalise]{cleveref}
\usepackage{xcolor}
\usepackage{multicol}
\usepackage{url}
\usepackage[utf8]{inputenc}
\usepackage{enumitem}
\usepackage{longtable}

\newcommand{\zo}{\{0,1\}}

\newcommand{\CSP}{\operatorname{CSP}}

\newcommand{\Q}{\mathbb{Q}}

\newcommand{\AND}{\mathbf{AND}}

\newcommand{\OR}{\operatorname{OR}}

\newcommand{\NRD}{\operatorname{NRD}}

\newcommand{\one}{{\bf 1}}
\newcommand{\wt}{\operatorname{wt}}

\newcommand{\Z}{\mathbb{Z}}

\newtheorem{theorem}{Theorem}
\numberwithin{theorem}{section}
\newtheorem{lemma}[theorem]{Lemma}
\newtheorem{claim}[theorem]{Claim}
\newtheorem{proposition}[theorem]{Proposition}

\newtheorem{conjecture}[theorem]{Conjecture}

\theoremstyle{definition}
\newtheorem{definition}[theorem]{Definition}

\allowdisplaybreaks

\title{Classification of Non-redundancy of Boolean Predicates of Arity 4}

\author{Joshua Brakensiek\thanks{University of California, Berkeley. Supported in part by NSF award DMS-2503280. Contact: \href{mailto:josh.brakensiek@berkeley.edu}{josh.brakensiek@berkeley.edu}} \and Venkatesan Guruswami\thanks{Simons Institute for the Theory of Computing and the University of California, Berkeley. Supported in part by a Simons Investigator award and NSF award CCF-2211972. Contact: \href{mailto:venkatg@berkeley.edu}{venkatg@berkeley.edu}} \and Aaron Putterman\thanks{Harvard University. Supported in part by the Simons Investigator awards of Madhu Sudan and Salil Vadhan, and AFOSR award FA9550-25-1-0112. Contact: \href{mailto:aputterman@g.harvard.edu}{aputterman@g.harvard.edu}}}

\begin{document}

\maketitle

\begin{abstract}
Given a constraint satisfaction problem (CSP) predicate $P \subseteq D^r$, the non-redundancy (NRD) of $P$ is maximum-sized instance on $n$ variables such that for every clause of the instance, there is an assignment which satisfies all but that clause. The study of NRD for various CSPs is an active area of research which combines ideas from extremal combinatorics, logic, lattice theory, and other techniques. Complete classifications are known in the cases $r=2$ and $(|D|=2, r=3)$. 

In this paper, we give a near-complete classification of the case $(|D|=2, r=4)$. Of the 400 distinct non-trivial Boolean predicates of arity 4, we implement an algorithmic procedure which perfectly classifies 397 of them. Of the remaining three, we solve two by reducing to extremal combinatorics problems---leaving the last one as an open question. Along the way, we identify the first Boolean predicate whose non-redundancy asymptotics are non-polynomial.
\end{abstract}

\section{Introduction}

In the study of constraint satisfaction problems (CSPs), an often-studied question is to find succinct representations of an input CSP instance such that some essential property of the original instance is preserved. The most common such name for such a task is that of \emph{sparsification}, although many of these formulations have subtly different objectives. For example, one can seek efficient transformations of an instance such that the transformed instance is satisfiable if and only if the input instance is (even if one cannot efficiently determine if the instance is satisfiable in the first place)~\cite{dell2014Satisfiability,jansen2019optimal,chen2020BestCase,lagerkvist2020Sparsification,carbonnel2022Redundancy}. A seemingly unrelated line of work seeks to (sparsely) reweight the clauses of an input CSP instance such that the fraction of clauses satisfied by every assignment is approximately preserved~\cite{DBLP:conf/soda/Karger93,BK96,BSS09,ST11,FHH11,KK15,FK17,SY19,BZ20,CCPS21,KKTY21a,JLS22,KKTY21b,Lee23,KPS24,khanna2024near,khanna2025efficient, brakensiek2025tight}.

A common thread between these two lines of work is the core combinatorial notions of \emph{non-redundancy} (NRD). The systematic study of NRD (and a related concept called \emph{chain length}) was initiated by Bessiere, Carbonnel, and Katsirelos~\cite{bessiere2020Chain} towards understanding the query complexity of CSPs in the partial membership query model~\cite{bessiere2013constraint}. Later, Carbonnel~\cite{carbonnel2022Redundancy} observed the utility of NRD in the study of CSP kernelization and further conjectured that the non-redundancy of any CSP is (approximately) an upper bound\footnote{The reason non-redundancy can only be an upper bound is that tractable CSPs (such as 2-SAT) have trivial kernels, although one can still study their NRD as a combinatorial quantity.} on the kernelization of any CSP. Very recently, Brakensiek and Guruswami~\cite{brakensiek2025redundancy} gave a direct link between combinatorial non-redundancy and the construction of (weighted) CSP sparsifiers, showing that the approaches of these disparate communities are (for the most part) working towards the same fundamental questions, often with overlapping results.

An example of these overlapping results is in the systematic classification of the non-redundancy of certain families of CSP predicates. For example, both Bessiere, Carbonnel, and Katsirelos~\cite{bessiere2020Chain} and Butti and {\v Z}ivn{\'y}~\cite{BZ20} (see also \cite{FK17}) classified all predicates of arity $2$, that is predicates of the form $P \subseteq D^2$ for some finite domain $D$. Although the classifications by these two works were in rather different settings, the dichotomies were identical:\footnote{We present the classification as it appears in \cite{bessiere2020Chain}, however due to a different ``satisfiability convention'' the classification in \cite{BZ20} is subtlely different. In particular, Butti and {\v Z}ivn{\'y}~\cite{BZ20} show if $P$ has a gadget reduction to $\AND_2 := \{11\}$, then $P$ has quadratic sparsifiability; otherwise $P$ has near-linear sparsifiability. These differing conventions each appear in various works, but for simplicity of exposition, we stick with the convention used by Bessiere, Carbonnel, and Katsirelos~\cite{bessiere2020Chain} which is consistent with the CSP kernelization literature. See Brakensiek and Guruswami~\cite{brakensiek2025redundancy} for further discussion.} if $P$ can define via a gadget reduction (more precisely an fgpp-reduction~\cite{carbonnel2022Redundancy}) the predicate $\OR_2 := \{00,01,10\}$, then $P$ has quadratic non-redundancy (where the growth rate is measured by the number of variables $n$ in the instance). Otherwise, $P$ has linear non-redundancy.

In a similar vein, both Chen, Jansen, and Pieterse~\cite{chen2020BestCase} and Khanna, Putterman, and Sudan~\cite{khanna2025efficient} (in effect) fully classified the non-redundancy of every ternary Boolean predicate $P \subseteq \{0,1\}^3$. The gadget reduction is very simple. If $P$ lacks a gadget reduction to $\OR_2$, then $P$ has linear non-redundancy. Otherwise, if $P$ has a gadget reduction to $\OR_2$ but not to $\OR_3 := \{0,1\}^3 \setminus \{000\}$, then $P$ has quadratic non-redundancy. Otherwise, $P$ has cubic non-redundancy. In particular, every such predicate has as its non-redundancy a polynomial function of the number of variables. This led Khanna, Putterman, and Sudan~\cite{khanna2025efficient} to conjecture that for their notion of CSP sparsifiability, the dependence on the number of variables is always a polynomial with an integral exponent.

However, such ``integral optimism'' turns out to be far from the truth. Both Bessiere, Carbonnel, and Katsirelos~\cite{bessiere2020Chain} and Carbonnel~\cite{carbonnel2022Redundancy} gave examples of predicates with natural non-redundancy upper bounds which were not integral polynomial functions of the number of variables, although they could not give lower bounds which proved the non-polynomial dependence is necessary. However, a few years later, Brakensiek and Guruswami~\cite{brakensiek2025redundancy} identified an explicit predicate $P \subseteq \{0,1,2\}^3$ whose non-redundancy as a function of the number of variables $n$ lies between $\Omega(n^{1.5})$ and $O(n^{1.6}\log n)$. Other predicates with fractional exponents were identified by Brakensiek, Guruswami, Jansen, Lagerkvist, and Wahlstr{\"o}m~\cite{brakensiek2025Richness}, including showing for every rational number $q \ge 1$ there exists a finite arity $r$ and a finite domain $D$ such that some $P \subseteq D^r$ has non-redundancy growing as $\Theta(n^q)$.

However, all such examples involve domains $D$ of size at least three. In particular, there are \emph{no known examples} of Boolean predicates $P \subseteq \{0,1\}^r$ whose non-redundancy fails to be an integral polynomial, thus leaving the conjecture of Khanna, Putterman, and Sudan~\cite{khanna2025efficient} open in this setting. That said, this conjecture is still generally believed to be false in the Boolean setting due to the existence of a notorious predicate\footnote{A succinct description of this predicate is the representation of the symmetric group $S_3$ over $\operatorname{GL}(3)$ with the identity matrix removed~\cite{brakensiek2025redundancy}.} $\mathsf{BoolBCK} \subseteq \{0,1\}^9$ independently identified by many different research groups~\cite{chen2020BestCase,lagerkvist2020Sparsification,khanna2025efficient}, all of which derive from a non-Boolean ternary predicate identified by Bessiere, Carbonnel, and Katsirelos~\cite{bessiere2020Chain}---see \cite{brakensiek2025redundancy,brakensiek2025Richness} for further discussion and analysis.

\subsection{Our Contributions} 
Note, however, that there is a large gap between the fully classified arity-$3$ predicates and this suspected-fractional predicate of arity $9$. The focus of this paper is to systematically understand Boolean predicates of arity $4$. In particular, using a computer search, we enumerate all $400$ non-trivial\footnote{Note that non-redundancy is unchanged under certain operations of the predicate, such as negating a bit or swapping two coordinates. This is why we are able to work with $400$ predicates instead of $2^{2^4}=65536$.} Boolean predicates $P \subseteq \{0,1\}^4$. To analyze these predicates, we implemented various techniques for both lower-bounding and upper-bounding the non-redundancy of $P$. In particular, for the lower bound, we checked the largest $r$ such that $P$ has a gadget reduction to $\OR_r$. For the upper bound, we used lattice-based techniques \cite{chen2020BestCase,khanna2025efficient} to identify low-degree polynomial representations of $P$. With only these two techniques, we were able to perfectly classify the non-redundancy of $397$ of the $400$ predicates, giving an integral exponent for each of them. The other three identified predicates are as follows (numbers derive from the table in \Cref{app:results}).
\begin{align*}
R_{181} &= \{0000, 0001, 0010, 0100, 0111, 1000, 1011, 1101, 1111\}\\
R_{299} &= \{ 0000, 0001, 0010, 0011, 0100, 0101, 1000, 1110, 1111\}\\
R_{317} &= \{ 0000, 0001, 0010, 0011, 0100, 0101, 1000, 1001, 1110\}.
\end{align*}
For each $R \in \{R_{181}, R_{299}, R_{317}\}$, our initial techniques show that $\NRD(R, n) \in [\Omega(n^2), O(n^3)]$, so further techniques are needed to resolve the non-redundancy of these predicates. For, $R_{317}$, we prove that $\Theta(n^3)$ is the correct answer. However, the lower bound does not following from a gadget reduction to $\OR_3$ but rather from the structure of \emph{linear 3-uniform hypergraphs}; see \Cref{subsec:317}.

For $R_{299}$, prove that $\NRD(R_{299}, n)$ does \emph{not} grow as a polynomial function of $n$, the first Boolean example of this phenomenon. This is also the first example of a predicate whose non-redundancy is asymptotically $n^{s-o(1)}$, where $s$ is a rational number. 
More precisely, we prove the following.
\begin{theorem}
$\NRD(R_{299}, n) \in \left[\frac{n^3}{2^{O(\sqrt{\log(n)})}}, \frac{n^3}{2^{\Omega(\log^*(n))}}\right]$.
\end{theorem}

The proof proceeds by showing that the non-redundancy of $R_{299}$ is closely linked to the optimal answer to the Rusza-Szemeredi problem of packing induced matchings in a graph. Such a method was used in~\cite{bessiere2020Chain} to upper bound the non-redundancy of a predicate $\mathsf{BCK} \subseteq \{0,1,2\}^3$, but this upper bound approach was later established to not be tight~\cite{brakensiek2025redundancy,brakensiek2025Richness}. In contrast, we identify a predicate in which the Rusza-Szemeredi analysis is (essentially) tight. See \Cref{subsec:299} and \Cref{thm:Pred299} for further details.

For $R_{181}$, we could not get a an improved bound on its non-redundancy. However, show $R_{181}$ is closely related to a \emph{symmetric} Boolean predicate of arity 5 which also remains unresolved by known methods, see \Cref{subsec:181} for further discussion.

\subsection{Outline} 
In \Cref{sec:prelim}, we formally define non-redundancy as well as give many known properties of non-redundancy (and related notions such as conditional non-redundancy). In \Cref{sec:experiment}, we describe in detail the procedure we used to systematically classify the 400 Boolean predicates of arity 4. In \Cref{sec:exceptions}, we discuss in detail the three exceptional predicates. In \Cref{sec:conclusion}, we conclude the paper with some thoughts and open questions. In \Cref{app:results}, we provide a table with the full 400-predicate classification.

\section{Preliminaries}\label{sec:prelim}

Our notation is adapted from Brakensiek and Guruswami~\cite{brakensiek2025redundancy}. We define a \emph{predicate} (or \emph{relation}) to be a subset $P \subseteq D^r$, where $D$ is the \emph{domain} of the predicate and $r$ is the \emph{arity} of the predicate. In most situations, we assume that $D = \{0,1\}$---that is, $P$ is a \emph{Boolean} predicate. We say that a predicate is \emph{trivial} if $P = \emptyset$ or $P = D^r$, otherwise the predicate is \emph{non-trivial}.

Given an arity $r$ predicate $P$, an \emph{instance} of $\CSP(P)$ is an $r$-uniform hypergraph $H := (V, E)$ where $V$ is the set of \emph{variables} of the instance and $E \subseteq V^r$ is the \emph{constraints}. In general, a constraint may use the same variable multiple times, however we usually assume that the instance is $r$-partite that the variables $V$ have a partitioning $V = V_1 \sqcup V_2 \sqcup \cdots \sqcup V_r$ such that $E \subseteq V_1 \times V_2 \times \cdots \times V_r$. We can even assume that $|V_1| = \cdots = |V_r|$ (assuming $|V|$ is a multiple of $r$). Such assumptions only affects quantities like the non-redundancy of the predicate up to a factor of $O_r(1)$.

An \emph{assignment} to $H$ is a map $\psi: V \to D$, where $D$ is the domain of the predicate $P$. A constraint $(x_1, \hdots, x_r) \in E$ is $P$-\emph{satisfied} (or just \emph{satisfied}) by $\psi$ if $(\psi(x_1), \hdots, \psi(x_r)) \in P$. Otherwise, the constraint is $P$-\emph{unsatisfied} (or just \emph{unsatisfied}). We now define the key notion of non-redundancy.

\begin{definition}
An instance $H := (V, E)$ is a \emph{non-redundant} instance of $\CSP(P)$ if for every $e \in E$ there is an assignment $\psi_e : V \to D$ such that $e' \in E$ is satisfied by $\psi_e$ if and only if $e \neq e'$. We let $\NRD(P, n)$ denote the maximum value of $|E|$ among all non-redundant instances $(V, E)$ of $\CSP(P)$ with $|V| = n$.
\end{definition}

As an example, if we let $\OR_r := \{0,1\}^r \setminus \{0^r\}$, then $\NRD(\OR_r, n) = \binom{n}{r}$ \cite{FK17,carbonnel2022Redundancy,khanna2025efficient} as the complete $r$-uniform hypergraph on $n$ vertices is non-redundant: for each $S \in \binom{[n]}{r}$, assign $\psi_S(v) = \one[v \not\in S]$.

\subsection{Conditional Non-redundancy}

Another crucial notion formally defined by Brakensiek and Guruswami~\cite{brakensiek2025redundancy} (although implicit in the work of Bessiere, Carbonnel, and Katsirelos~\cite{bessiere2020Chain}) is the notion of \emph{conditional} non-redundancy.

\begin{definition}
Given predicates $P \subsetneq Q \subseteq D^r$, we say that an $r$-partite hypergraph $H := (V, E)$ is a \emph{conditionally non-redundant} instance of $\CSP(P \mid Q)$ if for every $e \in E$ there is an assignment $\psi_e : V \to D$ (called a witness) such that $e' \in E$ is $P$-satisfied by $\psi_e$ if and only if $e \neq e'$ and further $e$ is $Q \setminus P$-satisfied by $\psi_e$. 
We let $\NRD(P \mid Q, n)$ denote the maximum value of $|E|$ among all conditionally non-redundant instances $(V, E)$ of $\CSP(P \mid Q)$ with $|V| = n$. Note that $\NRD(P \mid Q, n) \le \NRD(P, n)$, with equality when $Q = D^r$.
\end{definition}

We make use of the following fact about conditional non-redundancy.

\begin{proposition}[Triangle Inequality \cite{brakensiek2025redundancy}]\label{prop:NRD-facts}
Let $P \subseteq Q \subseteq R \subseteq D^r$ be predicates, and let $n$ be a positive integer. Then,
$\NRD(P \mid R, n) \le \NRD(P \mid Q, n) + \NRD(Q \mid R, n).$
\end{proposition}

We also make use of the following conditional non-redundancy trick due to Brakensiek et al.~\cite{brakensiek2025Richness}, where in some circumstances we can reduce the non-redundancy of a prediate of arity $r$ to a conditional NRD problem of arity $r-1$. We need the trick in more generality than their use case, so we give a full proof.

\begin{proposition}\label{prop:Magnus-trick}
Assume $\{0,1\} \subseteq D$. Let $R \subseteq D^r$ be a predicate such that $R = (Q \times \{0\}) \cup (P \times \{1\})$ with $P \subseteq Q \subseteq D^{r-1}$. Then,
\[
    \NRD(P \mid Q, n / 2) \cdot \Omega(n) \le \NRD(R, n) \le \NRD(P \mid Q, n ) \cdot O(n) + \NRD(Q \times \{0,1\}).
\]
\end{proposition}

\begin{proof}
We first prove the lower bound on $\NRD(R, n)$. Assume $n$ is even, and let $X \subseteq [n/2]^3$ be a conditionally non-redundant instance of $\CSP(P \mid Q)$, with witnessing assignments $\psi_e : [n/2] \to D$ for all $e \in X$. We claim that $Y := X \times \{n/2+1, \hdots, n\}$ is a non-redundant instance of $\CSP(R)$. To see why, fix a clause $e := (x_1,x_2,x_3) \in X$ and point $x_4 \in \{n/2+1,\hdots, n\}$. And consider $f := (x_1,x_2,x_3,x_4) \in Y$. Our witness $\psi_f : [n] \to D$ is defined by $\psi_f\upharpoonright [n/2] = \psi_e$, $\psi_f(x_4) = 1$ and $\psi_f(y) = 0$ for all $y \in \{n/2+1, \hdots, n\} \setminus \{y\}$. Since $R = (P \times \{1\}) \cup (Q \times \{0\})$, we have that $\psi_f$ satisfies all clauses of $Y$ except $f$.

For the upper bound, first note that by \Cref{prop:NRD-facts}, we have that $\NRD(R, n) \le \NRD(R \mid Q \times \{0,1\}, n) + \NRD(Q \times \{0,1\}, n)$. Thus, it suffices to prove that $\NRD(R \mid Q \times \{0,1\}, n) \le \NRD(P \mid Q, n ) \cdot O(n)$. To see why, let $Y \subseteq [n]^4$ be a non-redundant instance of $\CSP(R \mid Q \times \{0,1\})$, and let $Y' \subseteq Y$ be a sub-instance such that $|Y'| \ge |Y|/n$ and all clauses $e \in Y'$ have the same last coordinate $x \in [n]$. Let $Z := \{(x_1,x_2,x_3) : (x_1,x_2,x_3,x) \in Y\}$, so $|Z| = |Y'| \ge |Y|/n$. We claim that $Z$ is a conditionally non-redundant instance of $\CSP(P \mid Q)$, proving $|Y| \le |Z| n \le \NRD(P \mid Q, n) \cdot n$, as desired.

For each $f := (x_1,x_2,x_3,x) \in Y$, we have a witness $\psi_f : [n] \to D$ such that $\psi_f(f) \in (Q \setminus P) \times \{0\}$, so $\psi_f(x) = 0$. Thus, the condition that $\psi_f(f') \in R$ for all other $f' \in Y$ implies $\psi_f(f') \in P \times \{0\}$. Therefore, if we use $\psi_f$ as the witness for $(x_1,x_2,x_3) \in Z$, we have that $\psi_f(e) \in P$ unless $e = (x_1,x_2,x_3)$. Thus, $Z$ is a conditionally non-redundant instance of $\CSP(P \mid Q)$.
\end{proof}

\subsection{Balanced Predicates and Polynomial Representations}\label{subsec:poly}

In general, the best known\footnote{It is a major open question to establish if a Mal'tsev embedding is indeed the only reason a predicate can have linear (or even near-linear) non-redundancy~\cite{chen2020BestCase,bessiere2020Chain,lagerkvist2020Sparsification,carbonnel2022Redundancy,brakensiek2025redundancy,brakensiek2025Richness}.} method for testing for establishing that $\NRD(P, n) = O(n)$ for some predicate $P \subseteq D^r$ is checking if $P$ has what is known as a \emph{Malt'sev embedding}~\cite{chen2020BestCase,lagerkvist2020Sparsification,bessiere2020Chain}. However, due to a recent result of Brakensiek et al.~\cite{brakensiek2025Richness}, for Boolean predicates $P$, such criteria is equivalent to the much simpler criteria of checking whether $P$ is \emph{balanced} (e.g., \cite{chen2020BestCase}).

\begin{definition}[e.g., \cite{chen2020BestCase}]\label{def:balanced}
A Boolean predicate $P \subseteq \{0,1\}^r$ is \emph{balanced} if for all odd integers $m \ge 1$ and all $t_1, \hdots, t_m \in P$, we have that
\[
    t_1 - t_2 + t_3 - \cdots + t_m \in \{0,1\}^r \implies t_1 - t_2 + t_3 - \cdots + t_m \in P,
\]
where addition and subtraction is done over $\mathbb Z^r$.
\end{definition}

Given a Boolean predicate $P$, determining whether $P$ is balanced can be done using a lattice-based algorithm (e.g., \cite{chen2020BestCase,khanna2025efficient}), see \Cref{subsec:lattice} for more details. However, \emph{verifying} whether $P$ is balanced can be done much more efficiently using a linear polynomial certificate. In particular, Chen, Jansen, and Pieterse (Theorems 4.10 and 4.11 of \cite{chen2020BestCase}) show that $P \in \{0,1\}^r$ is balanced if and only if there exists integers $c_0, c_1, \hdots, c_r, q$ such that
\[
    (x_1, \hdots, x_r) \in P \iff c_1x_1 + \cdots + c_rx_r \equiv c_0 \mod q.
\]
This is equivalent to a Mal'tsev embedding of $P$ in $\mathbb Z/q\mathbb Z$ (see Theorem 7.6 of \cite{chen2020BestCase}).

More generally, a method of establishing $\NRD(P, n) = O(n^d)$ for some $d \ge 1$ is to find a modulus $q$ and coefficients $\{c_S \in \mathbb Z : S \subseteq [r], |S| \le d\}$ such that
\[
    (x_1, \hdots, x_r) \in P \iff \sum_{\substack{S \subseteq [r]\\|S| \le d}} c_S \prod_{i \in S} x_i \equiv 0 \mod q.
\]
See \cite{jansen2019optimal,chen2020BestCase,lagerkvist2020Sparsification} for a proof of this fact\footnote{To be truly pedantic, the notion of sparsification considered in \cite{jansen2019optimal,chen2020BestCase,lagerkvist2020Sparsification} is not formally equivalent to non-redundancy; however as noted by (e.g., \cite{brakensiek2025Richness}) the proof methods extend to the non-redundancy setting with no modification.} and further discussion.

\section{Experiment Methodology}\label{sec:experiment}

The primary contribution of this paper is a comprehensive classification of the non-redundancy of Boolean predicates of arity 4. In general, the techniques build on the previous classifications of Boolean predicates of arity 3~\cite{chen2020BestCase,khanna2025efficient}; however due to the shear number of predicates of arity 4 (400 total), our entire verification pipeline needs to be automated. In this section, we describe the high-level technical details of implementing this classification algorithm. The full classification is available in \Cref{app:results} with sufficient detail that hand-verification is possible. At the highest level, our classification algorithm (implemented\footnote{A copy of our code is available at \url{https://github.com/jbrakensiek/Arity-4-NRD}. In terms of verifying correctness, \Cref{app:results} includes sufficient details to make that possible by hand.} in Python and SageMath~\cite{sagemath,passagemath}) consists of four components.

\begin{enumerate}
\item Enumeration of Boolean predicates of arity 4 with symmetry-breaking.
\item Calculation of NRD lower bounds using OR projections.
\item Calculation of NRD upper bounds using lattice-based methods.
\item Certifying NRD upper bound by extracting a polynomial equation.
\end{enumerate}

We now discuss the primary technical ideas behind each of these steps.

\subsection{Enumeration of Boolean Predicates of Arity 4}

Enumerating all Boolean predicates of arity 4 is the most straightforward (and standard) component of our algorithm. Given that there are only $2^{2^4} = 65536$ possible predicates, a brute-force enumeration is practical. However, it is easy to observe that the non-redundancy of a predicate $P \subseteq \{0,1\}^4$ is invariant under swaps of the coordinates as well as bit-flipping individual coordinates. For example, if $P = \{1000,0100\}$, bit-flipping the first coordinate and swapping the middle two coordinates yields the predicate $P' = \{0000,1010\}$ with the same non-redundancy. Thus, we can substantially reduce the number of predicates by picking the lexicographically first predicate in each equivalence class.

Although we implemented our symmetry-breaking entirely by hand, the total number of distinct nontrivial Boolean predicates we found, 400, precisely matches previous enumerations of such predicates in the literature by Hast~\cite{hast2005beating} in the context of the Max CSP problem and Austrin, H{\aa}stad and Martinsson~\cite{austrin2026usefulness} in the context of the Promise CSP problem. These 400 predicates are enumerated from 0 to 399 in \Cref{app:results}. For example, in our encoding of the predicates, the lexicographically earliest representation of $P= \{1000,0100\}$ is $P'' = \{0000,0011\}$ which is predicate number 3 in the table.

\subsection{NRD Lower Bound: OR Projections}

A well-known non-redundancy bound is that the predicate $\OR_k := \{0,1\}^k \setminus \{0^k\}$ satisfies $\NRD(\OR_k, n) = \Theta_k(n^k)$~\cite{dell2014Satisfiability,FK17,BZ20,chen2020BestCase,lagerkvist2020Sparsification,carbonnel2022Redundancy,khanna2025efficient}. Such an NRD lower bound can be transferred to an arbitrary Boolean predicate $P \subseteq \{0,1\}^r$ using a suitable gadget reduction. Following the notation of Khanna, Putterman, and Sudan~\cite{khanna2025efficient}, we say that $\OR_k$ is a \emph{projection} of $P \subseteq \{0,1\}^r$ if there exists literals $\ell_1, \hdots, \ell_r \in \{0, 1, x_1, \hdots, x_k, \bar{x}_1, \hdots, \bar{x}_k\}$ such that 
\begin{align}
    \OR_k(x_1, \hdots, x_k) = P(\ell_1, \hdots, \ell_r),\label{eq:or-projection}
\end{align}
where we interpret the predicates as Boolean functions $\OR_k : \{0,1\}^k \to \{0,1\}$ and $P : \{0,1\}^r \to \{0,1\}$ in the standard manner. In this case, we can deduce that $\NRD(P, n) = \Omega(\NRD(\OR_k, n)) = \Omega_k(n^k)$. These ``OR projections'' are a special case of a much more general family of methods known as functionally guarded primitive positive (fgpp) reductions~\cite{carbonnel2022Redundancy,brakensiek2025Richness}.

In our setting, we have $r = 4$ and $k \in \{1, 2, 3, 4\}$. Thus, for our 400 predicates, we can efficiently enumerate over all choices of literals $\ell_1, \hdots, \ell_4$ and check if \Cref{eq:or-projection} is satisfied. In \Cref{app:results}, we provide an explicit choice of literals certifying to our claimed lower bound; note that such a certificate need not be unique in general.

\subsection{NRD Upper Bound: Lattice Methods}\label{subsec:lattice}

Next, we discuss our algorithmic techniques for upper-bounding the non-redundancy of a given predicate $P \subseteq \{0,1\}^4$. As a warm up, we discuss testing whether $\NRD(P, n) = O(n)$. In practice, this balanced condition is easy to check using a Lattice algorithm. First, consider map $e_1 : \{0,1\}^r \to \Z^{r+1}$ defined by $e_1(x_1, \hdots, x_r) = (1, x_1, \hdots, x_r).$
We can then define a lattice $\Lambda_{P,1}$ corresponding to a predicate $P \subseteq \{0,1\}^r$ as follows.
\[
    \Lambda_{P,1} := \Bigl\{\sum_{t \in P} z_t e_1(t) \mid \forall t \in P, z_t \in \mathbb Z\Bigr\}.
\]
We can then observe the following fact (implicit in previous works~\cite{chen2020BestCase,lagerkvist2020Sparsification,khanna2025efficient,brakensiek2025Richness}).

\begin{proposition}\label{prop:lattice}
    $P$ is balanced iff for all $t \in \{0,1\}^r$, $e_1(t) \in \Lambda_{P,1}$ implies $t \in P$.
\end{proposition}
\begin{proof}
For all $t' \in \{0,1\}^r$, The condition $e_1(t') \in \Lambda_{P,1}$ implies there exists integers $z^{t'}_t \in \mathbb Z$ such that $e_1(t') = \sum_{t \in P} z^{t'}_t e_1(t).$
Since the first coordinate of $e_1(t')$ is always $1$, we have that $\sum_{t \in P} z^{t'}_t = 1$. Let $m = \sum_{t \in P} |z^{t'}_t|$, which is odd. Then, there exists a choice of $t_1, \hdots, t_m \in P$, such that each $t \in P$ appears exactly $|z^{t'}_t|$ times further if $z^{t'}_t \ge 1$, then $t$ appears only in odd positions and if $z^{t'}_t \le -1$, then $t$ appears only in even positions. In other words, we have that $e_1(t') = e_1(t_1) - e_1(t_2) + e_1(t_3) - \cdots + e_1(t_m).$
Thus in particular, $t' = t_1 - t_2 + t_3 - \cdots + t_m$. Therefore, if $P$ is balanced, we have that $t' \in P$, so $e_1(t') \in \Lambda_P$ implies $t' \in P$.

Conversely, if $P$ is not balanced, then there are $t_1, \hdots, t_m \in P$ such that $t' := t_1 - t_2 + \cdots + t_m \in \{0,1\}^r \setminus P$. Then, it is clear that $e_1(t') = e_1(t_1) - e_1(t_2) + \cdots + e_1(t_m) \in \Lambda_P$. Thus, $e_1(t') \in \Lambda_P$ does not imply $t' \in P$.
\end{proof}

Using \Cref{prop:lattice}, we have a straightforward algorithm for checking if our predicate $P$ is balanced. We use the \textsf{IntegralLattice} library of SageMath~\cite{sagemath,passagemath} to build the lattice $\Lambda_P$. More precisely, we build a matrix $M$ whose rows are $e_1(t)$ for all $t \in P$. We compute the Hermite Normal Form of this matrix and extract the nonzero rows to get the basis $B$ of $\Lambda_P$. Using $B$, we build an \textsf{IntegralLattice} library which support membership queries. Thus, we can test for $t \in \{0,1\}^r$ whether $e_1(t) \in \Lambda_P$. If the test succeeds iff $t \in P$, then we know that $P$ is balanced and so $\NRD(P, n) = O(n)$.

\subsubsection{Higher Degrees} 
So far, we have discussed a method of checking if $\NRD(P, n) = O(n)$. However, as discussed in \Cref{subsec:poly} we can certify that $\NRD(P, n) = O(n^d)$ using a degree $d$ polynomial.  More precisely, for any $d\ge 1$, let $\binom{r}{\le d} := \sum_{i=0}^d \binom{r}{i}$ and consider the mapping $e_d : \{0,1\}^r \to \{0,1\}^{\binom{r}{\le d}}$ defined by
\[
    e_d(x_1, \hdots, x_r) = \left(\prod_{i \in S} x_i : \forall S \subseteq [r], |S| \le d\right).
\]
Note that $e_1$ coincides with our previous definition. For any $d \ge 1$ and predicate $P \subseteq \{0,1\}^d$, we can then build a lattice $\Lambda_{P,d}$ generated by $\{e_d(t): t \in P\}$. We then have the following analogue of \Cref{prop:lattice}

\begin{proposition}\label{prop:lattice-deg}
If for all $t \in \{0,1\}^r$, $e_d(t) \in \Lambda_{P,d}$ implies $t \in P$, then $\NRD(P, n) = O(n^d)$.
\end{proposition}
\begin{proof}
Let $R := \Lambda_{P,d} \cap \{0,1\}^{\binom{r}{\le d}}$. Since $e_d(t) \in R$ for all $t \in P$, we have that $\Lambda_{P,d} \subseteq \Lambda_{R,1}$. Since $R \subseteq \Lambda_{P,d}$, we have that $\Lambda_{R,1} \subseteq \Lambda_{P,d}$. Thus, $\Lambda_{R,1} = \Lambda_{P,d}$. Thus, $R = \Lambda_{R,1} \cap \{0,1\}^{\binom{r}{\le d}}$, so $R$ is balanced by \Cref{prop:lattice}. Therefore, there exists a finite $q$ such that and coefficients $c_S \in \mathbb Z/q\mathbb Z$ for all $S \subseteq [r], |S| \le d$ such that
\[
    x \in R \iff \sum_{\substack{S \subseteq [r]\\|S| \le d}} c_S x_S \equiv 0 \mod q.
\]
In particular, for all $t \in \{0,1\}^r$, we have that $e_d(t) \in R$ (that is, $t \in P$) iff
\[
\sum_{\substack{S \subseteq [r]\\|S| \le d}} c_S \prod_{i \in S} t_i \equiv 0 \mod q.
\]
In other words, $P$ is the zeros of a degree $d$ polynomial over $\mathbb Z/q\mathbb Z$. Therefore, by \cite{jansen2019optimal,chen2020BestCase,lagerkvist2020Sparsification}, we have that $\NRD(P, n) = O(n^d)$.
\end{proof}

Using \Cref{prop:lattice-deg}, we may test for arbitrary bounds of the form $O(n^d)$ by testing when \Cref{prop:lattice-deg} holds for all $d \in [r]$. These degree bounds for all 400 predicates appear in \Cref{app:results}.

\subsection{Extracting a Polynomial Certificate}

Unlike the OR projections, the lattice algorithm does not immediately yield a succinct certificate of the claimed non-redundancy upper bound. However, by adapting the methods of Chen, Jansen, and Pieterse~\cite{chen2020BestCase} and Khanna, Putterman, and Sudan~\cite{khanna2025efficient}, we can extract a suitable polynomial for which hand-verification becomes feasible. We concern ourselves with the task of verifying whether $P$ is balanced---the extension of this technique to higher degrees follows with minimal adaptation.

Recall that $\Lambda_{P,1} \subseteq \Z^{r+1}$ is the lattice generated by $\{e_1(t) : t \in P\}$. We let $\Lambda_{P,1} \otimes \Q \subseteq \Q^{r+1}$ denote the vector space over $\mathbb Q$ generated by $\Lambda_{P,1}$. Recall that to verify if $P$ is balanced, we must check for every $t \in \{0,1\}^d \setminus P$ that $e_1(t) \not\in \Lambda_{P,1}$. This can happen for two reasons: either (1) $e_1(t) \not\in \Lambda_{P,1} \otimes \Q$, or (2) $e_1(t) \in \Lambda_{P,1} \otimes \Q$ but $e_1(t)\not\in \Lambda_{P,1}$.

The first case is equivalent to the existence of a vector $y \in \Q^{d+1}$ such that $\langle y, \Lambda_{P,1}\rangle = 0$ but $\langle y, e_1(t))\rangle \neq 0$. By picking a suitable positive integer $N$ such that $Ny$ has integral entries as well as a positive integer $q$ which does not divide $\langle Ny, e_1(t)\rangle$, then the equation $\langle Ny, e_1(t')\rangle \equiv 0 \mod q$ includes all $t' \in P$ but excludes $t' = t$.

In the second case, we can construct a similar map, by looking at the dual lattice
\[
    \Lambda_{P,1}^* := \{y \in \Q^{d+1}  \mid \forall x \in \Lambda_{P,1}, \langle y, x\rangle \in \mathbb Z\},
\]
and finding a map $y \in \Lambda_{P,1}^*$ such that $\langle y, e_1(t)\rangle \not\in \mathbb Z$. If $\langle y, e_1(t)\rangle = p/q$ with $p,q \in \mathbb Z$ and pick positive $N$ such that $Ny$ has integral entries then $Nq \langle y, e_1(t')\rangle \equiv 0 \mod Nq$ for all $t' \in P$ but $Nq \langle y, e_1(t)\rangle \not\equiv 0\mod Nq$.

In both cases, we have found a linear polynomial which has all elements of $P$ as zeros but not $t$. In Khanna, Putterman, and Sudan~\cite{khanna2025efficient}, the authors aggregate these linear polynomials for all $t \in \{0,1\}^r \setminus P$ to get a system\footnote{Equivalently, this system can be represented as a single linear equation over an Abelian group.} of linear equations witnessing $P$ is balanced. However, as predicated by \cite{chen2020BestCase}, we found (with a suitable brute-force search over small moduli), we could find a single linear equation excluding all $t \in \{0,1\}^r \setminus P$ lying in Case (1) as well as a single linear equation excluding all $t \in \{0,1\}^r \setminus P$ lying in Case (2). Given the freedom in Case (1) of choosing the modulus $q$, we could pick the two equations to be over relatively prime moduli, ensuring by the Chinese Remainder Theorem can be combined into a single equation. This extracted linear equation for each of the 400 predicates is stated in \Cref{app:results}.

\section{The Three Exceptional Predicates}\label{sec:exceptions}

As previously mentioned, the algorithmic techniques in \Cref{sec:experiment} prefectly classified the non-redundancy of 397 of the 400 predicates. In this section, we study the three exceptional predicates, namely Predicates 181, 299, and 317, providing tight theoretical analyses of their $\NRD$ for two of them. We start with the analysis of Predicates 317 and 299 which we provide tight bounds for, and conclude with the discussion of Predicate 181, which we leave as an open problem.

\subsection{Predicate 317}\label{subsec:317}

Predicate 317 has satisfying assignments 
\[
R_{317} = \{ 0000, 0001, 0010, 0011, 0100, 0101, 1000, 1001, 1110\}.
\]
From \Cref{app:results}, we see that $x \in \{0,1\}^4$ satisfies $x \in R_{317}$ if and only if $x_{1}x_2+x_1x_3+x_2x_3+x_1x_2x_4 \equiv 0\mod 3.$
Therefore, $\NRD(R_{317}, n) = O(n^3)$. We thus seek to prove that $\NRD(R_{317}, n) = \Omega(n^3)$. First, observe that
\[
R_{317}  = \{000, 001, 010, 100, 111\} \times \{0,1\} - \{1111\}.
\]

Thus, setting $P_{317} = \{000, 001, 010, 100\}, Q_{317} = P_{317} \cup \{111\}$, we have
$R_{317} = (Q_{317} \times \{0\}) \cup (P_{317} \times \{1\})$. 
By \cref{prop:Magnus-trick}, we know that $\NRD(P_{317} \mid Q_{317}, n / 2) \cdot \Omega(n) \le \NRD(R_{317}, n).$

\subsubsection{Lower Bounding Conditional Non-Redundancy}

We now proceed to a \emph{lower bound} on $\NRD(P_{317} \mid Q_{317}, n)$. For this, we use the following constructive property of linear hypergraphs:

\begin{theorem}[See, for instance \cite{deza1978intersection}]\label{thm:LHLB}
	There exist tripartite linear hypergraphs on vertex set $[n] \times [n] \times [n]$ with $\Omega(n^2)$ edges.
\end{theorem}

With this, we show:

\begin{lemma}\label{lem:pred317LB}
	$\NRD(P_{317} \mid Q_{317}, n) = \Omega(n^2)$.
\end{lemma}

\begin{proof}
We consider a tripartite linear hypergraph $H = (V, E)$ as guaranteed by \cref{thm:LHLB}, with $|E| = \Omega(n^2)$. We claim that this hypergraph is conditionally non-redundant, thus yielding the lemma. To see why this hypergraph is conditionally non-redundant, we consider any hyperedge $e \in E$, along with the assignment $\psi_e$ which maps $\psi_e(x) = 1$ if $x \in e$, and otherwise maps variables to $0$. Note then that $\psi_e(e) = 111$, as we desire. It remains only to show that $\psi_e(e') \in \{000, 001, 010, 100\}$ for any $e' \neq e$. This follows because $|e' \cap e| \leq 1$, and so at least two variables in $e'$ are mapped to $0$ by $\psi_e$. This concludes the proof.
\end{proof}

Thus, we are now able to prove the following:

\begin{theorem}\label{thm:Pred317}
	$\NRD(R_{317}, n) = \Theta \left ( n^3 \right )$.
\end{theorem}

\begin{proof}
	By \cref{prop:Magnus-trick}, we know that $\NRD(P_{317} \mid Q_{317}, n / 2) \cdot \Omega(n) \le \NRD(R_{317}, n).$ Plugging in \cref{lem:pred317LB} then yields $\NRD(R_{317}, n) = \Omega(n^3)$. Since $R_{317}$ has a cubic polynomial representation, we know that $\NRD(R_{317}, n) = O(n^3)$, completing the proof.
\end{proof}

\subsection{Predicate 299}\label{subsec:299}

Predicate 299 has satisfying assignments
\[
R_{299} = \{ 0000, 0001, 0010, 0011, 0100, 0101, 1000, 1110, 1111\} .
\]

In this section, we show the following theorem:

\begin{theorem}\label{thm:Pred299}
	$\NRD(R_{299}, n) \geq  \frac{n^3}{2^{O(\sqrt{\log(n)})}} $ and $\NRD(R_{299}, n)  \leq \frac{n^3}{2^{\Omega(\log^*(n))}} $.
\end{theorem}

Towards proving this theorem, observe that we can write 
\[
R_{299} = \{000, 001, 010, 100, 111\} \times \{0, 1\} - \{1001\}.
\]
Now, let us set $P_{299} = \{ 000, 001, 010, 111\}$ and $Q_{299} = P_{299} \cup \{100\}$. Immediately, we can observe that $R_{299} = (Q_{299} \times \{0\}) \cup (P_{299} \times \{1\})$. Thus, by \cref{prop:Magnus-trick}, we know that 
\[
\NRD(P_{299} \mid Q_{299}, n / 2) \cdot \Omega(n) \le \NRD(R_{299}, n) \le \NRD(P_{299} \mid Q_{299}, n ) \cdot O(n) + \NRD(Q_{299} \times \{0,1\}).
\]

Next, we show that $Q_{299} \times \{0,1\}$ has non-redundancy $O(n^2)$:

\begin{claim}\label{clm:QBoundPred299}
	$\NRD(Q_{299} \times \zo, n) = O(n^2)$.
\end{claim}

\begin{proof}
    $Q_{299} \times \zo$ corresponds to predicate 318 in \Cref{app:results}. The upper bound of $O(n^2)$ follows from the fact that for any $x \in \{0,1\}^4$, we have that $x \in Q_{299} \times \zo$ if and only if
$x_{1}x_2 + x_1x_3 + x_2x_3 \equiv 0 \mod 3.$\qedhere
\end{proof}

We make an additional observation that simplifies our notation going forward: because $\NRD$ is unchanged when permuting and negating variables in a predicate, we are free to re-define $P'_{299} = \{001, 011, 101, 110\}$ and  $Q'_{299} = P'_{299} \cup \{000\}$, while ensuring that $\NRD(P'_{299} \mid Q'_{299}, n) = \NRD(P_{299} \mid Q_{299}, n)$. Note that here we have simply negated the first bit, and then swapped the first bit with the third bit. Now, our goal is to bound $\NRD(P'_{299} \mid Q'_{299} , n)$.

\subsubsection{Upper Bounding the Conditional Non-Redundancy}

Before bounding the conditional non-redundancy of these predicates, we first recall the statement of Fox's triangle removal lemma \cite{fox2011new}:

\begin{theorem}[Fox's Triangle Removal Lemma, simplified]\label{thm:triangleRemoval}
	Let $G$ be a graph on $n$ vertices with at most $c n^2$ many triangles, where $c$ is an absolute constant. Then, there exists a set of $\frac{n^2}{2^{\Omega_c(\log^*(n))}}$ many edges that can be removed such that $G$ becomes triangle-free.
\end{theorem}

\begin{lemma}\label{lem:299UB}
	$\NRD(P'_{299} \mid Q'_{299}, n) \leq   \frac{n^2}{2^{\Omega(\log^*(n))}}$.
\end{lemma}

\begin{proof}
We show that $\NRD(P'_{299} \mid Q'_{299}, 3n) \geq \frac{n^2}{2^{\Omega(\log^*(n))}}$, from which the stated lemma follows. 

	Let $E = \{e_1 \dots e_m \} \subseteq [n] \times [n] \times [n]$ denote the constraints which are a conditionally non-redundant instance of $P'_{299} \mid Q'_{299}$ (at a constant factor loss in size, we may assume that that non-redundant instances are tripartite). That is, for each constraint $e_i \in E$, there exists a corresponding assignment $\psi_{e_i} \in \zo^n$ such that $\psi_{e_i}(e_i) = 000$, and for all other $e_j \neq e_i$, we have that $\psi_{e_i}(e_j) \in \{001, 011, 101, 110 \}$.
	
	Now, we consider the following construction: we create an auxiliary graph $G$ with vertex set $A = [n]$, $B = [n]$ and $C = [n]$, and for every edge $e_i = (u,v,w)$, we add an edge from $u \in A$ to $v \in B$, from $u \in A$ to $w \in C$, and from $v \in B$ to $w \in C$. We define these triangles $(u,v, w)$ to be \textbf{canonical triangles} and note that their number is $|E|$. As our goal is to bound $|E|$, we will show that one can remove all triangles in $G$ (and in particular all the canonical triangles) by deleting few edges, and that each edge on average cannot remove too many canonical triangles.
	
	For ease of analysis later, we will also use the following subgraphs: for a choice of $c \in [n]$, we let $E[c]$ denote the graph which contains edges $\{e-c: e \in \{e_1 \dots e_m\}, e[3] =c \}$ (in words, we are taking all constraints which have variable $c$ in the third position, and considering the graph they induce on $A \cup B$). We only include those vertices in this graph which touch one of these edges. 
	Importantly, we have the following two claims about these induced graphs:
	
	\begin{claim}\label{clm:matching}
		Let $E = \{e_1 \dots e_m \} \subseteq [n] \times [n] \times [n]$ denote the (ordered) constraints which form a witness to the non-redundancy of $P'_{299} \mid Q'_{299}$, and let $E[c]$ be defined as above for $c \in [n]$. Then, $E[c]$ is a matching.
	\end{claim}

	\begin{claim}\label{clm:boundedInducedIntersection}
		Let $E = \{e_1 \dots e_m \} \subseteq  [n] \times [n] \times [n]$ denote the (ordered) constraints which form a witness to the non-redundancy of $P'_{299} \mid Q'_{299}$, and let $E[c]$ be defined as above for $c \in [n]$. Then, for any $c \neq c' \in [n]$, $V(E[c])$ induces at most $2$ edges of $E[c']$.
	\end{claim}
	
	We defer the proofs of these claims for now, but, as we shall see, they suffice for deriving our ultimate bound on the size of our starting conditionally non-redundant instance. Towards this end, we first use the above claims to show that the graph $G$ has only $O(n^2)$ many triangles. To see why, consider any triple of vertices $a \in A, b \in B, c \in C$ which form a triangle (note that no triangle can have two vertices in any of $A, B, C$, as there are no internal edges in these parts). If $a, b, c$ form a triangle, this means that there must be edges $(a, b), (b, c), (a,c )$ in the graph $G$. There are two cases for us to consider:
	\begin{enumerate}
		\item Suppose that the edge $(a, b) \in E[c]$. Then, this trivially means that $(a, b), (b,c), (a,c) \in G$ which form the aforementioned canonical triangles.
		\item  Suppose that the edge $(a, b) \notin E[c]$. Then, it must be the case that there is some edge $(a, y) \in E[c]$, some edge $(x, b) \in E[c]$, and some $c'$ for which $(a, b) \in E[c']$. It may even be the case that $x=a$ and $y=b$ (i.e., $E[c']$ and $E[c]$ may have shared edges).
        Importantly, this means that the edge $(a, b)$ from $E[c']$ is \emph{induced} by $V(E[c])$. We call these \textbf{non-canonical triangles}.
	\end{enumerate}
	
	Now we bound the number of these triangles. First, we see that the number of canonical triangles in $G$ is bounded by $n^2$ because $\sum_{c \in [n]} |E[c]| \leq n \cdot \max_{c \in [n]} |E[c]| \le n^2$, as by \cref{clm:matching}, we know that each $E[c]$ is a matching (and so has at most $n$ edges).
	
	Next, we bound the number of non-canonical triangles. The key observation here, as remarked above, is that any non-canonical triangle contains an edge $(a,b) \in E[c']$ which is induced by $V(E[c])$. To bound the number of non-canonical triangles, we see (by \cref{clm:boundedInducedIntersection}):
	\[
	\sum_{c \in [n]} \sum_{c' \neq c \in [n]} \left |\{e \in E[c']: e \text{ induced by } V(E[c])\}\right | \leq \sum_{c \in [n]} \sum_{c' \neq c \in [n]} 2 \le 2n^2.
	\]
	
	So, for this graph $G$, we know that $G$ has at most $3n^2$ triangles. By Fox's triangle removal lemma (\cref{thm:triangleRemoval}), we then know that there exists a set of $\frac{n^2}{2^{\Omega(\log^*(n))}}$ many edges that can be removed from $G$ such that the resulting graph is triangle-free. 
	
	Importantly, this implies that removing $\frac{n^2}{2^{\Omega(\log^*(n))}}$ many edges from $G$ removes \emph{all} of the canonical triangles from the graph $G$. Let $E'$ denote this set of edges which are deleted from $G$. 
    We now analyze how many \emph{canonical triangles} are removed by the removal of each edge $e \in E'$, call it $\ell_e$.
	\begin{enumerate}
		\item Suppose that a deleted edge in $E'$ is of the form $(a, c)$ or $(b, c)$ (WLOG we consider $(a, c)$). In this case, the only canonical triangles which are lost are those which correspond to edges in $E[c]$ which contain the vertex $a$. But, because $E[c]$ is a matching (by \cref{clm:matching}), we know that there is at most one edge in $E[c]$ which contains $a$. Thus, there is at most one canonical triangle which is lost. Thus, we set $\ell_{(a, c)} = 1$. 
		\item Suppose a deleted edge in $E'$ is of the form $(a, b)$. In this case, a canonical triangle $(a, b, c)$ is lost if and only if $(a,b) \in E[c]$. Thus, the number of canonical triangles which is removed is exactly $\ell_{(a,b)}$: the number of choices of $c \in [n]$ such that $(a, b) \in E[c]$. 
	\end{enumerate}
	
	Now, by \cref{clm:boundedInducedIntersection}, we know that every pair of labels $c, c' \in [n]$ has at most two edges in common. Hence, 
	\[
	2 \cdot \binom{n}{2} \geq \sum_{e \in E'} \binom{\ell_e}{2} \geq |E'| \cdot \binom{\frac{1}{|E'|} \cdot \sum_{e \in E'} \ell_e}{2},
	\]
	where the second inequality above is simply using convexity. The above then implies that $|E'| \cdot \frac{(\sum_{e \in E'} \ell_e)^2}{|E'|^2} \leq 2 \cdot \binom{n}{2} \leq n^2$,
	or equivalently, that $\sum_{e \in E'} \ell_e \leq n \cdot \sqrt{|E'|}.$
	By the above, we know that $|E'| = \frac{n^2}{2^{\Omega(\log^*(n))}}$, and so this immediately implies that $\sum_{e \in E'} \ell_e \leq \frac{n^2}{2^{\Omega(\log^*(n))}}.$
	To conclude, we have shown that deleting the set of edges $E'$ removes all triangles from the graph $G$. At the same time, this process removes at most $\frac{n^2}{2^{\Omega(\log^*(n))}}$ canonical triangles from the graph. Together, this means that $G$ has at most $ \frac{n^2}{2^{\Omega(\log^*(n))}} $ many canonical triangles, and so the number of constraints in our starting non-redundant CSP is at most $\frac{n^2}{2^{\Omega(\log^*(n))}}$.
\end{proof}

Now we prove our auxiliary claims. We start by proving \cref{clm:matching}.

\begin{proof}[Proof of \cref{clm:matching}]
	As before, let the constraints in the non-redundant CSP instance be $E = \{e_1, \dots e_m\}$, and consider $E[c]$ for any $c \in [n]$. Now, for any edge $(a,b) \in E[c]$, we know that $(a, b, c)$ is a constraint in our CSP instance. Importantly, because our CSP instance is non-redundant, we know there must exist an assignment $\psi \in \zo^n$ such that  $\psi(a, b, c) = 000$, while for every other constraint $e'$, $\psi(e') \in \{ 001, 011, 101, 110 \}$. Now, in order for $\psi(a, b, c) = 000$, this means that all of $\psi(a), \psi(b), \psi(c) = 0$. Now, suppose for the sake of contradiction that there were any other edge in $E[c]$ which shared a vertex with $(a, b)$ (WLOG assume $a$), and denote this edge by $(a, b')$. Then, under the assignment $\psi$, we see that $\psi(a, b', c) = 0 \circ \psi(b') \circ 0 \notin \{ 001, 011, 101, 110 \}$ for any choice of $\psi(b')$. Thus, this instance could not be a witness to the non-redundancy of $P'_{299} \mid Q'_{299}$. So, in any non-redundant instance of $P'_{299} \mid Q'_{299}$, it must be the case that each vertex in $E[c]$ has degree $\leq 1$ (i.e., a matching).
\end{proof}

\begin{proof}[Proof of \cref{clm:boundedInducedIntersection}]
	Consider any $c, c' \in [n]$, and let us look at $E[c]$ and $E[c']$. Now, suppose for the sake of contradiction that there are three edges $e_1 = (a_1, b_1), e_2 = (a_2, b_2), e_3 = (a_3, b_3)$ from $E[c']$ which are induced by $V(E[c])$. Because there are three distinct edges, and $E[c']$ is a matching, it must be the case that there are three distinct left vertices $a_1, a_2, a_3$ and likewise three distinct right vertices  $b_1, b_2, b_3$.
	
	Now, let us consider the edge $e \in E[c]$ such that $a_1 \in e$, along with the assignment $\psi$ which is the non-redundant certificate for this edge $e$. By definition, it must then be the case that $\psi(a_1) = 0$. Importantly, this now forces the value of $\psi(c')$: indeed, conditioned on the first variable being $0$, the only value of the third variable which is in $P_{299}$ is $1$. Thus, we know that $\psi(c') = 1$.
	
	At the same time, we know that for every variable $x$ in $V(E[c]) - e$, it must be the case that $\psi(x) = 1$. This is because the only tuple which is in $P_{299}$ with final variable $0$ is $110$ (and by the above, we know that $\psi(c) = 0$, effectively forcing these other assignments). To reach a contradiction then, we know that because $a_1, a_2, a_3$ are distinct, and $b_1, b_2, b_3$ are distinct, one of $(a_2, b_2)$ or $(a_3, b_3)$ must \emph{not} intersect $e$ (WLOG let us say that it is $(a_3, b_3)$). But, because $(a_3, b_3)$ is induced by $V(E[c])$, this means that $a_3, b_3 \in  V(E[c])$, and so $\psi(a_3) = \psi(b_3) = 1$. At the same time, know that $\psi(c') = 1$, which means that $\psi(a_3, b_3, c') = 111$ which is \emph{not} in $P_{299}$. Thus, this violates the fact that we started with an NRD instance. 
	
	So, it must be that $V(E[c])$ induces at most $2$ edges from $E[c']$.
\end{proof}

\subsubsection{Lower Bounding the Conditional Non-Redundancy}

We now provide a lower bound on the conditional non-redundancy $\NRD(P_{299} \mid Q_{299})$. For our construction, we will make use of the fascinating \emph{Ruzsa-Szemeredi} graphs~\cite{ruzsa1978triple}.

\begin{definition}
	We say that a graph $G = (V, E)$ on $n$ vertices is an $(r, t)$-RS graph if $E$ can be partitioned in a set of $t$ matchings $M_1, \dots M_t$, each of size $r$ such that for every $i \in [t]$, $M_i$ is an \emph{induced} matching in $G$.
\end{definition}

The key fact we will use is the following. 
\begin{theorem}[\cite{behrend1946sets, fischer2002monotonicity}]\label{thm:goodRSExists}
	There exist bipartite $(r,t)$-RS graphs on vertex set $|L| = |R| = n$ with $t = \Omega(n)$ and $r \geq \frac{n}{2^{O(\sqrt{\log n})}}$.
\end{theorem}

With this in hand, we now show the following.
\begin{lemma}\label{lem:299LB}
	$\NRD(P_{299}  \mid Q_{299}, n) \ge \frac{n^2}{2^{O(\sqrt{\log n})}}$.
\end{lemma}

\begin{proof}
	First, we fix a (bipartite) $(r, t)$ RS graph on $2n$ vertices (vertex set $L, R$) with $t = \Omega(n)$ and $r = \frac{n}{2^{O(\sqrt{\log(n)})}}$ as guaranteed by \cref{thm:goodRSExists}. To build our NRD instance, we now add $t$ distinguished variables $z_1, \dots z_t$, along with $2n$ variables $x_1, \dots x_n, y_1, \dots y_n$ for the vertices in the graph. Now, for each matching $M_i: i \in [t]$, we add a constraint of $\{u, v, z_i\}$ 
    for $e = (u,v) \in M_i$ (where we understand that $R_{299}$ operates on $x_u, y_v, z_i$, where $e = (u,v)$). Thus, the total number of constraints in the resulting CSP (denote this CSP by $C$) is exactly $|E| = t \cdot r = \frac{n^2}{2^{O(\sqrt{\log n})}}.$
	
	It remains only to show that $C$ is a (conditionally) non-redundant instance. To see this, let us fix a single constraint $R_{299}(e \cup \{z_i\})$  for $e = (u,v) \in M_i$, with the understanding that $u \in L$, $v \in R$. We must show that there is an assignment $\psi: \{x_1, \dots x_n, y_1, \dots y_n, z_1, \dots z_t\} \rightarrow \zo$ such that $\psi(e \cup \{z_i\}) = 000$, while for every other constraint $\psi(e' \cup \{z_j\}) \in \{001, 011, 101, 110 \}$.
	
	For this, we consider the assignment $\psi$ such that
	\[
	\psi(z_i) = 0 \quad \quad \quad \quad \quad \quad \quad \quad \quad \quad \quad \quad \quad \quad \quad \quad \quad \forall j \neq i \in [t]: \psi(z_j) = 1,
	\]
	\[
	\psi(x_a)= \begin{cases}
		0 \text{ if } a = u \\
		1 \text{ if } a \in L \cap V(M_i) - \{ u\} \\
		0 \text{ if } a \in L - V(M_i)
	\end{cases}, \quad 	\psi(y_b)= \begin{cases}
		0 \text{ if } b = v \\
		1 \text{ if } b \in R\cap V(M_i) - \{ v\} \\
		0 \text{ if } b \in R - V(M_i)
	\end{cases}.
	\]
	By construction, we know that $\psi(e \cup \{z_i\}) = 000$. So, it remains only to show that $\psi(e' \cup \{z_j\}) \in \{001, 011, 101, 110 \}$ for every $e' \in M_j$. We do this by cases:
	\begin{enumerate}
		\item Suppose $e' \in M_i$, where $M_i$ is the same matching that contains $e$. Observe then that $e' \cap e = \emptyset$ as $M_i$ is a matching. This means that both vertices in $e'$ are given value $1$ by the assignment $\psi$, i.e., $\psi(e') = 11$. By construction, we know that $\psi(z_i) = 0$, so $\psi(e' \cup \{z_j\}) = 110$. 
		\item Suppose $e' \in M_j$ for $j \neq i$. Then, we know that $\psi(z_j) = 1$ by construction. Likewise, we know that $|e' \cap V(M_i)| \leq 1$, as $M_i$ was an \emph{induced} matching in the graph $G$ (which implies that every edge $e' \in G - M_i$ has at least one vertex outside of $V(M_i)$). Thus, $\psi(e') \in \{00, 01, 10\}$. Thus, $\psi(e' \cup \{z_j\}) \in \{001, 011, 101\}$, as we desire. \qedhere
	\end{enumerate}
\end{proof}

Thus, we are now able to prove \Cref{thm:Pred299}. 

\begin{proof}[Proof of \cref{thm:Pred299}.]
By \cref{prop:Magnus-trick}, we know that 
\[
\NRD(P_{299} \mid Q_{299}, n / 2) \cdot \Omega(n) \le \NRD(R_{299}, n) \le \NRD(P_{299} \mid Q_{299}, n ) \cdot O(n) + \NRD(Q_{299} \times \{0,1\}).
\]
Plugging in \cref{lem:299UB}, \cref{lem:299LB} and \cref{clm:QBoundPred299}, then yields the desired claim. 
\end{proof}

\subsection{Predicate 181 and Connections to Hypergraph Intersection Problems}\label{subsec:181}

Predicate 181 has satisfying assignments 
\[
R_{181} = \{0000, 0001, 0010, 0100, 0111, 1000, 1011, 1101, 1111\}.
\]

    From \Cref{app:results} 
    we know that $\NRD(R_{181}, n) \in [\Omega(n^2), O(n^3)]$. More precisely, for the lower bound, we can observe that $R_{181}(0,0,\bar{x}_1,\bar{x}_2) = \OR_2(x_1, x_2)$. For the upper bound, $R_{181}$ has a polynomial representation of \[
    2x_1x_2 + x_1x_3 + x_1x_4+x_2x_3+x_2x_4+x_3x_4+2x_1x_2x_4 \mod 3.
    \] 
    However, we are unable to further classify the non-redundancy of $R_{181}$.

\paragraph*{Connection to a Symmetric Arity $5$ Predicate.}
Consider the predicate $P_{181} = \{x \in \zo^5: \wt(x) \in \{0, 2, 3\} \}$. Immediately, we can observe that predicate $Q_{181}  = \{x \in \zo^5: \wt(x) \in \{0, 2, 3, 5\} \}$ satisfies $\NRD(Q_{181}, n) = O(n^2)$, as $Q_{181}$ is exactly the zeros of the degree polynomial 
$f(x) = |x| \cdot (|x| - 2) \mod 3.$

Thus, by \cref{prop:NRD-facts}, we know that $\NRD(P_{181} , n) \leq \NRD(P_{181} \mid Q_{181} , n) + \NRD(Q_{181}, n) \leq \NRD(P_{181} \mid Q_{181} , n) + O(n^2)$. So understanding the $\NRD$ of $P_{181}$ reduces to understanding the $\NRD$ of $P_{181} \mid Q_{181}$ (up to an additive error of $O(n^2)$).

A priori, it may seem mysterious that we mention this arity $5$ predicate $P_{181}$ in the discussion of Predicate 181. However, we can consider the restriction of $P_{181}$ that is achieved by forcing $x_4 = x_5$: this then yields satisfying assignments
\[
(P_{181})_{x_4 = x_5} = \{0000, 0001, 0110, 1100, 1010, 1001, 0101, 0011, 1110\}.
\]
Even further, if we now negate the last variable, we see that:
\[
(P_{181})_{x_4 = x_5, \neg x_4} = \{0001, 0000, 0111, 1101, 1011, 1000, 0100, 0010, 1111 \}
\]
\[
= \{0000, 0001, 0010, 0100, 0111, 1000, 1011, 1101, 1111 \} = R_{299}
\]

Thus, by using \cref{prop:NRD-facts}, we have:

\begin{claim}
    $\NRD(R_{181}, n) \leq \NRD(P_{181}, n) \leq \NRD(P_{181} \mid Q_{181}, n) + \NRD(Q_{181},n) \leq \NRD(P_{181} \mid Q_{181}, n) + O(n^2)$.
\end{claim}

Therefore, one promising approach towards bounding $\NRD(R_{181},n )$ is via an upper bound on $\NRD(P_{181}\mid Q_{181}, n)$.

\paragraph*{Connecting $\NRD(P_{181} \mid  Q_{181}, n)$ to Hypergraph Intersection Problems.} 
To better understand $\NRD(P_{181} \mid Q_{181}, n)$, we consider the hypergraph interpretation of this quantity: any hypergraph $H = (V, E)$ which constitutes a conditionally non-redundant instance for $\NRD(P_{181} \mid Q_{181}, n)$ must satisfy the following property: for each $e \in E$ (each $e$ is of size $5$) there is an assignment $\psi_e: V \rightarrow \zo$ such that $(\psi_e)(e) = 11111$, but for every other $e' \neq e \in E$, that $\wt(\psi_e(e')) \in \{0, 2, 3\}$. 

A natural approach towards building such non-redundant hypergraphs is to consider the set of assignments $\psi_e(x_i) = \mathbf{1}[i \in e]$, i.e., where $\psi_e$ only maps variables in $e$ to $1$, and all other variables are mapped to $0$. In this setting, we can observe that a hypergraph is conditionally non-redundant \emph{if and only if} for every pair of hyperedges $e, e'$, $|e \cap e'| \in \{0, 2, 3\}$. This motivates a natural \emph{extremal hypergraph intersection problem}:
    \emph{What is the maximum number of hyperedges in a hypergraph $H = (V, E)$ on $n$ vertices of arity $5$ such that for every pair of hyperedges $e, e'$, $|e \cap e'| \in \{0, 2, 3\}$?}

This style of question is studied exactly in the work of Deza, Erd{\"o}s, and Frankl \cite{deza1978intersection}, where they showed that any such hypergraph must \emph{necessarily have} $O(n^2)$ many hyperedges. In turn, this immediately implies that the ``canonical'' set of assignments $\psi_e$ defined above cannot create conditionally non-redundant instances of size larger than $n^2$. Note that an upper bound of $O(n^3)$ is standard by the linear algebra/polynomial method of Ray-Chaudhuri, Wilson~\cite{RayChaudhuriWilson1975} (and prime modulus versions in \cite{FranklWilson1981}). Deza {\it et al}~\cite{deza1978intersection} provide a powerful refinement of this approach to derive a tighter $O(n^2)$ bound, which we can hope to emulate in our CSP setting.

\paragraph*{Barriers in Bounding $\NRD(P_{181} \mid Q_{181}, n)$.}
Note that the discussion above not provide a general bound on the maximum size non-redundant instance however, as there is no condition which forces the assignments $\psi_e$ to be the degenerate family of indicator assignments for $e \in E$. Nevertheless, it does raise the question of whether the \emph{techniques} from \cite{deza1978intersection} can be modified to work in this new conditionally non-redundant setting. At a high level, the key insight from \cite{deza1978intersection} is the following:

\begin{claim}\label{clm:coreStructure}
    Let $H = (V, E)$ be a hypergraph of arity $5$ with $|e \cap e'| \in \{0, 2, 3\}$. For any vertex $v \in V$, let $H_v = (V, E_v)$ be the hypergraph with all hyperedges containing $v$. If $|E_v| = \omega(n)$, then there exists another vertex $u \in V$ such that for all $e \in E_v$, $(u,v) \subseteq e$.
\end{claim}

Essentially, \cite{deza1978intersection} shows that as soon as a hypergraph $H$ with appropriate intersection sizes reaches a critical density, it must be the case that the hypergraphs $H_v$ all have cores of $2$ shared vertices. This gives \cite{deza1978intersection} sufficient structural information to then derive their $O(n^2)$ upper bound. Unfortunately, in our setting, an analog of \cref{clm:coreStructure} is false. Indeed, a candidate analog would be the following:

\begin{conjecture}[False]\label{clm:coreStructureAnalog}
    Let $H = (V, E)$ be a hypergraph of arity $5$ which is conditionally non-redundant for $P_{181} \mid Q_{181}$. For any vertex $v \in V$, let $H_v = (V, E_v)$ be the hypergraph with all hyperedges containing $v$. If $|E_v| = \omega(n)$, then there exists another vertex $u \in V$ such that for all $e \in E_v$, $(u,v) \subseteq e$.
\end{conjecture}

As an example, let us partition the vertex set $V - \{v \}$ into two parts of equal size, denoted by $V_1, V_2$. Now, inside $V_1$, we select a vertex $u_1$ arbitrarily, and within $V_2$ we select a vertex $u_2$ arbitrarily.  On the remaining vertices $V_1 - \{u_1\}$, we build a hypergraph $H_1$ of arity $3$, with size $\Omega(n^2)$ such that for all $f, f' \in H_1$, $|f \cap f'| \leq 1$ (note such a construction is possible, see for instance \cite{deza1978intersection}).\footnote{Formally, this is called a linear hypergraph.} Likewise, we build $H_2$ on $V_2 - \{u_2\}$. We now let our hypergraph be $E_v = \{\{u_1,v\} \cup e: e \in H_1 \} \cup \{\{u_2,v\} \cup e: e \in H_2 \}$. Importantly, we can observe that this hypergraph \emph{does not} satisfy the structural condition of \cref{clm:coreStructureAnalog}. Nevertheless, we claim it is conditionally non-redundant for $P_{181} \mid Q_{181}$. To see this, consider any hyperedge $\{(u_1, v), e\} \in E$, assuming (WLOG) that $e \in H_1$. Now, consider the assignment $\psi_e$ which maps variables in $e$ to $1$, maps $v, u_1,u_2$ to $1$, and maps every other variable to $0$. We see that:
\begin{enumerate}
    \item $\psi_e(\{(u_1, v), e\}) = 11111$.
    \item For any $\{u_1,v\} \cup e': e' \in H_1$, $u_1, v$ both map to $1$, and because $|e' \cap e| \leq 1$, we see that $\wt(\{u_1,v\} \cup e') \geq 2$ and $\wt(\{u_1,v\} \cup e') \leq 3$, as at most one other variable in $e'$ maps to $1$.
    \item For any $\{u_2,v\} \cup e': e' \in H_2 $, $u_2, v$ both map to $1$, and $e'$ maps to $0$ (recall, these variables are disjoint from $V_1$). Thus, $\wt(\{u_2,v\} \cup e') = 2$.
\end{enumerate}
Thus, this is indeed a non-redundant instance of $P_{181} \mid Q_{181}$, thereby showing that there is no possible characterization along the lines of \cref{clm:coreStructureAnalog}.

\section{Conclusion and Open Questions}\label{sec:conclusion}

In this paper, we classified the non-redundancy of 399 of 400 Boolean predicates of arity 4. We leave the remaining case $R_{181}$ as our main open question. The most notable case we classified, $R_{299}$, is the first example of a predicate Boolean or otherwise whose non-redundancy is provably not of the form $\Theta(n^{\beta})$ for some rational number $\beta$. Does there exist a Boolean predicate (of arity $r \ge 4$) whose non-redundancy is $\Theta(n^{\beta})$ for some $\beta \in \mathbb Q \setminus \mathbb N$. We know that $\beta$ cannot lie in the range $(r-1, r)$~\cite{chen2020BestCase,carbonnel2022Redundancy,khanna2025efficient}. Does there exist any Boolean predicate whose non-redundancy lies in the range $[\omega(n), o(n^2)]$?

\bibliographystyle{alphaurl}
\bibliography{references}

\raggedbottom
\pagebreak

\appendix

\section{Full table of results}\label{app:results}

In this appendix, we present the complete classification of Boolean predicates of arity 4. We give a brief explanation of each column of the table.

\begin{itemize}
\item \textbf{Predicate.} The specific predicate $P \subseteq \{0,1\}^4$ written in binary in the following ordering:
\[
    (0000, 0001, 0010, 0011, 0100, 0101, 0110, 0111, 1000, 1001, 1010, 1011, 1100, 1101, 1110, 1111).
\]

\item \textbf{OR.} The largest $k \in [4]$ such that some projection of $P$ is equal to $\OR_k$.

\item \textbf{OR Certificate.} A specific projection of $P$ certifying the OR degree. 

\item \textbf{Deg.} TA minimal degree of a polynomial representation of $P$.

\item \textbf{Polynomial.} A polynomial representation of $P$. We use $x_S$ as shorthand for $\prod_{i \in S} x_i$.

\end{itemize}

\addtolength{\tabcolsep}{-0.4em} 
\footnotesize
\begin{longtable}{ccccccc}
Num. & Predicate & OR & OR Cert. & Deg & Polynomial\\
\hline
\endhead
0 & 1000000000000000 & 1 & $(0,0,0,\bar{x}_1)$ & 1 & $x_{4}+x_{3}+x_{2}+x_{1}\mod 5$\\
1 & 1000000000000001 & 1 & $(0,0,0,\bar{x}_1)$ & 1 & $x_{4}+x_{3}+x_{2}+2x_{1}\mod 5$\\
2 & 1000000100000000 & 1 & $(0,0,0,\bar{x}_1)$ & 1 & $x_{4}+x_{3}+5x_{2}+3x_{1}\mod 7$\\
3 & 1001000000000000 & 1 & $(0,0,0,\bar{x}_1)$ & 1 & $x_{4}+6x_{3}+2x_{2}+2x_{1}\mod 7$\\
4 & 1001000000000100 & 1 & $(0,0,0,\bar{x}_1)$ & 1 & $x_{4}+4x_{3}+2x_{2}+2x_{1}\mod 5$\\
5 & 1001000000000110 & 1 & $(0,0,0,\bar{x}_1)$ & 1 & $6x_{4}+6x_{3}+x_{2}+5x_{1}\mod 12$\\
6 & 1001000000001000 & 2 & $(\bar{x}_1,\bar{x}_1,\bar{x}_2,\bar{x}_2)$ & 2 & $x_{4}+x_{3}+x_{34}+x_{2}+2x_{1}+x_{14}+x_{13}\mod 3$\\
7 & 1001000000001001 & 1 & $(0,0,0,\bar{x}_1)$ & 1 & $x_{4}+4x_{3}+2x_{2}+3x_{1}\mod 5$\\
8 & 1001010000000000 & 1 & $(0,0,0,\bar{x}_1)$ & 1 & $x_{4}+4x_{3}+4x_{2}+3x_{1}\mod 5$\\
9 & 1001010000000010 & 1 & $(0,0,0,\bar{x}_1)$ & 1 & $x_{4}+4x_{3}+4x_{2}+2x_{1}\mod 5$\\
10 & 1001010000100000 & 2 & $(\bar{x}_1,\bar{x}_2,\bar{x}_1,\bar{x}_2)$ & 2 & $x_{4}+x_{3}+x_{34}+x_{2}+x_{24}+x_{1}+x_{13}\mod 3$\\
11 & 1001010000101000 & 2 & $(\bar{x}_1,\bar{x}_1,\bar{x}_2,\bar{x}_2)$ & 2 & $x_{4}+x_{3}+x_{34}+x_{2}+x_{24}+x_{1}+x_{13}+x_{12}\mod 3$\\
12 & 1001010000101001 & 1 & $(0,0,0,\bar{x}_1)$ & 1 & $x_{4}+2x_{3}+2x_{2}+x_{1}\mod 3$\\
13 & 1001010001000000 & 1 & $(0,0,0,\bar{x}_1)$ & 1 & $x_{4}+4x_{3}+4x_{2}+4x_{1}\mod 5$\\
14 & 1001010001000010 & 1 & $(0,0,0,\bar{x}_1)$ & 1 & $4x_{4}+2x_{3}+2x_{2}+2x_{1}\mod 6$\\
15 & 1001011000000000 & 1 & $(0,0,0,\bar{x}_1)$ & 1 & $3x_{4}+3x_{3}+3x_{2}+4x_{1}\mod 6$\\
16 & 1001011001000000 & 2 & $(\bar{x}_1,\bar{x}_1,\bar{x}_2,\bar{x}_2)$ & 2 & $x_{4}+x_{3}+x_{34}+2x_{2}+x_{1}+x_{14}+2x_{13}+2x_{12}\mod 3$\\
17 & 1001011001100000 & 2 & $(\bar{x}_1,\bar{x}_2,\bar{x}_1,\bar{x}_2)$ & 2 & $x_{4}+x_{3}+x_{34}+2x_{2}+2x_{1}\mod 3$\\
18 & 1001011001101000 & 2 & $(\bar{x}_1,\bar{x}_1,\bar{x}_2,\bar{x}_2)$ & 2 & $x_{4}+x_{3}+3x_{34}+x_{2}+3x_{24}+3x_{23}+x_{1}+3x_{14}+3x_{13}+3x_{12}\mod 5$\\
19 & 1001011001101001 & 1 & $(0,0,0,\bar{x}_1)$ & 1 & $3x_{4}+3x_{3}+3x_{2}+3x_{1}\mod 6$\\
20 & 1100000000000000 & 1 & $(0,0,\bar{x}_1,0)$ & 1 & $x_{3}+x_{2}+x_{1}\mod 5$\\
21 & 1100000000000010 & 2 & $(\bar{x}_1,\bar{x}_1,\bar{x}_1,\bar{x}_2)$ & 2 & $x_{3}+x_{2}+2x_{23}+x_{1}+x_{14}+2x_{13}+2x_{12}\mod 3$\\
22 & 1100000000000011 & 1 & $(0,0,\bar{x}_1,0)$ & 1 & $x_{3}+x_{2}+x_{1}\mod 3$\\
23 & 1100001000000000 & 2 & $(0,\bar{x}_1,\bar{x}_1,\bar{x}_2)$ & 2 & $x_{3}+x_{2}+x_{24}+x_{23}+2x_{1}+x_{13}+x_{12}\mod 3$\\
24 & 1100001000000001 & 2 & $(0,\bar{x}_1,\bar{x}_1,\bar{x}_2)$ & 2 & $x_{3}+x_{2}+x_{24}+x_{23}+x_{1}+2x_{13}+2x_{12}\mod 3$\\
25 & 1100001000000010 & 2 & $(0,\bar{x}_1,\bar{x}_1,\bar{x}_2)$ & 2 & $x_{3}+x_{2}+x_{24}+x_{23}+x_{1}+2x_{12}\mod 3$\\
26 & 1100001000010000 & 2 & $(0,\bar{x}_1,\bar{x}_1,\bar{x}_2)$ & 2 & $x_{3}+x_{2}+x_{24}+x_{23}+2x_{1}+2x_{14}+x_{13}+x_{12}\mod 3$\\
27 & 1100001000010010 & 2 & $(0,\bar{x}_1,\bar{x}_1,\bar{x}_2)$ & 2 & $x_{3}+2x_{2}+2x_{24}+2x_{1}+2x_{14}+x_{13}\mod 3$\\
28 & 1100001000100000 & 2 & $(0,\bar{x}_1,\bar{x}_1,\bar{x}_2)$ & 2 & $x_{3}+x_{2}+x_{24}+x_{23}+x_{1}+x_{14}+x_{13}\mod 3$\\
29 & 1100001000100001 & 2 & $(0,\bar{x}_1,\bar{x}_1,\bar{x}_2)$ & 2 & $x_{3}+2x_{2}+2x_{24}+2x_{1}+2x_{14}\mod 3$\\
30 & 1100001000100100 & 2 & $(0,\bar{x}_1,\bar{x}_1,\bar{x}_2)$ & 2 & $x_{3}+x_{34}+x_{2}+x_{23}+x_{1}+x_{14}+x_{13}\mod 3$\\
31 & 1100001000100101 & 2 & $(0,\bar{x}_1,\bar{x}_1,\bar{x}_2)$ & 2 & $x_{3}+x_{2}+x_{24}+x_{23}+x_{1}+x_{14}+x_{13}+2x_{12}\mod 3$\\
32 & 1100001000101000 & 2 & $(0,\bar{x}_1,\bar{x}_1,\bar{x}_2)$ & 2 & $x_{3}+2x_{2}+2x_{24}+2x_{1}+2x_{14}+2x_{12}\mod 3$\\
33 & 1100001000101001 & 2 & $(0,\bar{x}_1,\bar{x}_1,\bar{x}_2)$ & 2 & $x_{3}+x_{2}+x_{24}+3x_{23}+x_{1}+2x_{14}+3x_{13}+3x_{12}\mod 5$\\
34 & 1100001100000000 & 1 & $(0,0,\bar{x}_1,0)$ & 1 & $x_{3}+4x_{2}+2x_{1}\mod 5$\\
35 & 1100001100100000 & 2 & $(\bar{x}_1,0,\bar{x}_1,\bar{x}_2)$ & 2 & $x_{3}+x_{2}+x_{23}+x_{1}+x_{14}+x_{13}+2x_{12}\mod 3$\\
36 & 1100001100100100 & 2 & $(\bar{x}_1,0,\bar{x}_1,\bar{x}_2)$ & 2 & $x_{3}+2x_{2}+x_{1}+x_{14}+x_{13}+2x_{12}\mod 3$\\
37 & 1100001100101000 & 2 & $(\bar{x}_1,0,\bar{x}_1,\bar{x}_2)$ & 2 & $x_{3}+x_{2}+x_{23}+2x_{1}+2x_{14}\mod 3$\\
38 & 1100001100110000 & 1 & $(0,0,\bar{x}_1,0)$ & 1 & $x_{3}+2x_{2}+2x_{1}\mod 3$\\
39 & 1100001100111000 & 2 & $(1,\bar{x}_1,\bar{x}_1,\bar{x}_2)$ & 2 & $x_{3}+x_{34}+2x_{2}+4x_{24}+2x_{23}+3x_{1}+4x_{14}+x_{13}\mod 5$\\
40 & 1100001100111100 & 1 & $(0,0,\bar{x}_1,0)$ & 1 & $3x_{3}+3x_{2}+3x_{1}\mod 6$\\
41 & 1110000000000000 & 2 & $(0,0,\bar{x}_1,\bar{x}_2)$ & 2 & $x_{34}+x_{2}+x_{1}+2x_{12}\mod 3$\\
42 & 1110000000000001 & 2 & $(0,0,\bar{x}_1,\bar{x}_2)$ & 2 & $x_{34}+x_{2}+x_{1}\mod 3$\\
43 & 1110000000000100 & 2 & $(0,0,\bar{x}_1,\bar{x}_2)$ & 2 & $x_{34}+x_{2}+2x_{1}+2x_{14}+x_{12}\mod 3$\\
44 & 1110000000000101 & 2 & $(0,0,\bar{x}_1,\bar{x}_2)$ & 2 & $x_{34}+2x_{2}+2x_{24}+2x_{1}+2x_{14}+2x_{13}+x_{12}\mod 3$\\
45 & 1110000000000110 & 2 & $(0,0,\bar{x}_1,\bar{x}_2)$ & 2 & $x_{34}+x_{2}+x_{1}+x_{14}+x_{13}\mod 3$\\
46 & 1110000000000111 & 2 & $(0,0,\bar{x}_1,\bar{x}_2)$ & 2 & $x_{34}+x_{2}+2x_{1}+2x_{14}+2x_{13}+x_{12}\mod 3$\\
47 & 1110000000001000 & 2 & $(0,0,\bar{x}_1,\bar{x}_2)$ & 2 & $x_{34}+x_{2}+2x_{1}+2x_{14}+2x_{13}\mod 3$\\
48 & 1110000000001001 & 2 & $(0,0,\bar{x}_1,\bar{x}_2)$ & 2 & $x_{34}+x_{2}+x_{1}+x_{14}+x_{13}+x_{12}\mod 3$\\
49 & 1110000000001100 & 2 & $(0,0,\bar{x}_1,\bar{x}_2)$ & 2 & $x_{34}+x_{2}+x_{24}+x_{23}+2x_{1}+2x_{14}\mod 3$\\
50 & 1110000000001101 & 2 & $(0,0,\bar{x}_1,\bar{x}_2)$ & 2 & $x_{34}+x_{2}+2x_{1}+2x_{13}\mod 3$\\
51 & 1110000000001110 & 2 & $(0,0,\bar{x}_1,\bar{x}_2)$ & 2 & $x_{34}+x_{2}+x_{1}+x_{12}\mod 3$\\
52 & 1110000100000000 & 2 & $(0,0,\bar{x}_1,\bar{x}_2)$ & 2 & $x_{34}+x_{2}+x_{23}+2x_{1}+2x_{13}+x_{12}\mod 3$\\
53 & 1110000100000001 & 2 & $(0,0,\bar{x}_1,\bar{x}_2)$ & 2 & $x_{34}+x_{2}+x_{23}+x_{1}+2x_{12}\mod 3$\\
54 & 1110000100000100 & 2 & $(0,0,\bar{x}_1,\bar{x}_2)$ & 2 & $x_{34}+x_{2}+x_{23}+2x_{1}+2x_{14}+x_{12}\mod 3$\\
55 & 1110000100000101 & 2 & $(0,0,\bar{x}_1,\bar{x}_2)$ & 2 & $x_{34}+x_{2}+x_{23}+x_{1}+x_{14}+x_{13}\mod 3$\\
56 & 1110000100000110 & 2 & $(0,0,\bar{x}_1,\bar{x}_2)$ & 2 & $x_{34}+2x_{2}+x_{1}+x_{14}+x_{13}+2x_{12}\mod 3$\\
57 & 1110000100001000 & 2 & $(0,0,\bar{x}_1,\bar{x}_2)$ & 2 & $x_{34}+x_{2}+x_{23}+x_{1}+x_{14}+x_{13}+x_{12}\mod 3$\\
58 & 1110000100001001 & 2 & $(0,0,\bar{x}_1,\bar{x}_2)$ & 2 & $x_{34}+2x_{2}+x_{1}+x_{14}+x_{13}\mod 3$\\
59 & 1110000100001100 & 2 & $(0,0,\bar{x}_1,\bar{x}_2)$ & 2 & $x_{34}+x_{2}+x_{23}+x_{1}+x_{12}\mod 3$\\
60 & 1110000100001110 & 2 & $(0,0,\bar{x}_1,\bar{x}_2)$ & 2 & $x_{34}+x_{2}+x_{23}+2x_{1}+2x_{13}\mod 3$\\
61 & 1110000100010000 & 2 & $(0,0,\bar{x}_1,\bar{x}_2)$ & 2 & $x_{34}+x_{2}+x_{23}+x_{1}+x_{13}\mod 3$\\
62 & 1110000100010001 & 2 & $(0,0,\bar{x}_1,\bar{x}_2)$ & 2 & $x_{34}+x_{2}+x_{1}+x_{12}\mod 2$\\
63 & 1110000100010100 & 2 & $(0,0,\bar{x}_1,\bar{x}_2)$ & 2 & $x_{34}+x_{2}+x_{24}+x_{1}+x_{14}+2x_{12}\mod 3$\\
64 & 1110000100010110 & 2 & $(0,0,\bar{x}_1,\bar{x}_2)$ & 2 & $x_{34}+x_{2}+x_{23}+x_{1}+x_{14}\mod 3$\\
65 & 1110000100011000 & 2 & $(0,0,\bar{x}_1,\bar{x}_2)$ & 2 & $x_{34}+x_{2}+3x_{23}+x_{1}+2x_{14}+x_{13}+3x_{12}\mod 5$\\
66 & 1110000100011001 & 2 & $(0,0,\bar{x}_1,\bar{x}_2)$ & 2 & $x_{34}+x_{2}+x_{23}+x_{1}+x_{14}+x_{12}\mod 3$\\
67 & 1110000100011100 & 2 & $(0,0,\bar{x}_1,\bar{x}_2)$ & 2 & $x_{34}+x_{2}+x_{23}+2x_{1}\mod 3$\\
68 & 1110000100011110 & 2 & $(0,0,\bar{x}_1,\bar{x}_2)$ & 2 & $x_{34}+x_{2}+x_{1}\mod 2$\\
69 & 1110010000000000 & 2 & $(0,0,\bar{x}_1,\bar{x}_2)$ & 2 & $x_{34}+x_{2}+2x_{24}+x_{1}\mod 3$\\
70 & 1110010000000010 & 2 & $(0,0,\bar{x}_1,\bar{x}_2)$ & 2 & $x_{34}+x_{2}+2x_{24}+x_{1}+x_{14}+x_{13}\mod 3$\\
71 & 1110010000000011 & 2 & $(0,0,\bar{x}_1,\bar{x}_2)$ & 2 & $x_{34}+x_{2}+2x_{24}+x_{23}+x_{1}\mod 3$\\
72 & 1110010000010000 & 2 & $(0,0,\bar{x}_1,\bar{x}_2)$ & 2 & $x_{34}+x_{2}+2x_{24}+x_{1}+x_{14}+2x_{12}\mod 3$\\
73 & 1110010000010001 & 2 & $(0,0,\bar{x}_1,\bar{x}_2)$ & 2 & $x_{34}+x_{2}+2x_{24}+x_{1}+x_{14}\mod 3$\\
74 & 1110010000010010 & 2 & $(0,0,\bar{x}_1,\bar{x}_2)$ & 2 & $x_{34}+x_{2}+2x_{24}+x_{23}+x_{1}+x_{14}\mod 3$\\
75 & 1110010000010011 & 2 & $(0,0,\bar{x}_1,\bar{x}_2)$ & 2 & $x_{34}+x_{2}+2x_{24}+x_{1}+x_{13}\mod 3$\\
76 & 1110010000011000 & 2 & $(0,0,\bar{x}_1,\bar{x}_2)$ & 2 & $x_{34}+x_{2}+2x_{24}+x_{1}+x_{13}+x_{12}\mod 3$\\
77 & 1110010000011010 & 2 & $(0,0,\bar{x}_1,\bar{x}_2)$ & 2 & $x_{34}+2x_{2}+x_{24}+2x_{1}+2x_{12}\mod 3$\\
78 & 1110010000011011 & 2 & $(0,0,\bar{x}_1,\bar{x}_2)$ & 2 & $x_{34}+x_{2}+x_{24}+x_{1}\mod 2$\\
79 & 1110010000100000 & 2 & $(0,0,\bar{x}_1,\bar{x}_2)$ & 2 & $x_{34}+x_{2}+2x_{24}+x_{1}+2x_{13}\mod 3$\\
80 & 1110010000100001 & 2 & $(0,0,\bar{x}_1,\bar{x}_2)$ & 2 & $x_{34}+x_{2}+2x_{24}+x_{23}+x_{1}+x_{14}+2x_{13}\mod 3$\\
81 & 1110010000100100 & 2 & $(0,0,\bar{x}_1,\bar{x}_2)$ & 2 & $x_{34}+x_{2}+2x_{24}+2x_{1}+x_{13}+x_{12}\mod 3$\\
82 & 1110010000100110 & 2 & $(0,0,\bar{x}_1,\bar{x}_2)$ & 2 & $x_{34}+2x_{2}+x_{24}+2x_{1}+x_{13}+x_{12}\mod 3$\\
83 & 1110010000100111 & 2 & $(0,0,\bar{x}_1,\bar{x}_2)$ & 2 & $x_{34}+x_{2}+x_{24}+x_{1}+x_{13}+x_{12}\mod 2$\\
84 & 1110010000101000 & 2 & $(0,0,\bar{x}_1,\bar{x}_2)$ & 2 & $x_{34}+x_{2}+x_{24}+x_{1}+x_{13}\mod 2$\\
85 & 1110010000101001 & 2 & $(0,0,\bar{x}_1,\bar{x}_2)$ & 2 & $x_{34}+2x_{2}+x_{24}+2x_{1}+x_{13}+2x_{12}\mod 3$\\
86 & 1110011000000000 & 2 & $(0,0,\bar{x}_1,\bar{x}_2)$ & 2 & $x_{34}+2x_{2}+x_{24}+x_{23}+x_{1}+x_{12}\mod 3$\\
87 & 1110011000000001 & 2 & $(0,0,\bar{x}_1,\bar{x}_2)$ & 2 & $x_{34}+2x_{2}+x_{24}+x_{23}+2x_{1}+2x_{13}\mod 3$\\
88 & 1110011000000100 & 2 & $(0,0,\bar{x}_1,\bar{x}_2)$ & 2 & $x_{34}+2x_{2}+x_{24}+x_{23}+2x_{1}+2x_{13}+x_{12}\mod 3$\\
89 & 1110011000000101 & 2 & $(0,0,\bar{x}_1,\bar{x}_2)$ & 2 & $x_{34}+2x_{2}+3x_{24}+3x_{23}+x_{1}+x_{13}+4x_{12}\mod 5$\\
90 & 1110011000000110 & 2 & $(0,0,\bar{x}_1,\bar{x}_2)$ & 2 & $x_{34}+2x_{2}+x_{24}+x_{23}+x_{1}+2x_{12}\mod 3$\\
91 & 1110011000000111 & 2 & $(0,0,\bar{x}_1,\bar{x}_2)$ & 2 & $x_{34}+2x_{2}+x_{24}+x_{23}+x_{1}+x_{14}+x_{13}+x_{12}\mod 3$\\
92 & 1110011000001000 & 2 & $(0,0,\bar{x}_1,\bar{x}_2)$ & 2 & $x_{34}+2x_{2}+x_{24}+x_{23}+x_{1}+x_{14}+x_{13}\mod 3$\\
93 & 1110011000001001 & 2 & $(0,0,\bar{x}_1,\bar{x}_2)$ & 2 & $x_{34}+2x_{2}+x_{24}+x_{23}+x_{1}\mod 3$\\
94 & 1110011000001100 & 2 & $(0,0,\bar{x}_1,\bar{x}_2)$ & 2 & $x_{34}+2x_{2}+x_{24}+x_{23}+2x_{1}+2x_{14}+2x_{12}\mod 3$\\
95 & 1110011000001110 & 2 & $(0,0,\bar{x}_1,\bar{x}_2)$ & 2 & $x_{34}+2x_{2}+x_{24}+x_{23}+2x_{1}+2x_{14}+2x_{13}+2x_{12}\mod 3$\\
96 & 1110011000010000 & 2 & $(0,0,\bar{x}_1,\bar{x}_2)$ & 2 & $x_{34}+2x_{2}+x_{24}+x_{23}+2x_{1}\mod 3$\\
97 & 1110011000010001 & 2 & $(0,0,\bar{x}_1,\bar{x}_2)$ & 2 & $x_{34}+2x_{2}+3x_{24}+3x_{23}+2x_{1}+2x_{13}+2x_{12}\mod 5$\\
98 & 1110011000010100 & 2 & $(0,0,\bar{x}_1,\bar{x}_2)$ & 2 & $x_{34}+2x_{2}+x_{24}+x_{23}+x_{1}+x_{14}+x_{12}\mod 3$\\
99 & 1110011000010101 & 2 & $(0,0,\bar{x}_1,\bar{x}_2)$ & 2 & $x_{34}+2x_{2}+x_{24}+x_{23}+x_{1}+x_{13}+2x_{12}\mod 3$\\
100 & 1110011000010110 & 2 & $(0,0,\bar{x}_1,\bar{x}_2)$ & 2 & $x_{34}+2x_{2}+x_{24}+x_{23}+2x_{1}+x_{12}\mod 3$\\
101 & 1110011000011000 & 2 & $(0,0,\bar{x}_1,\bar{x}_2)$ & 2 & $x_{34}+2x_{2}+x_{24}+x_{23}+x_{1}+x_{13}\mod 3$\\
102 & 1110011000011001 & 2 & $(0,0,\bar{x}_1,\bar{x}_2)$ & 2 & $x_{34}+2x_{2}+x_{24}+x_{23}+2x_{1}+2x_{12}\mod 3$\\
103 & 1110011000011100 & 2 & $(0,0,\bar{x}_1,\bar{x}_2)$ & 2 & $x_{34}+2x_{2}+3x_{24}+3x_{23}+2x_{1}+2x_{14}+x_{12}\mod 5$\\
104 & 1110011000011110 & 2 & $(0,0,\bar{x}_1,\bar{x}_2)$ & 2 & $x_{34}+3x_{2}+2x_{24}+2x_{23}+3x_{1}+3x_{14}+3x_{13}+4x_{12}\mod 5$\\
105 & 1110011100000000 & 2 & $(0,0,\bar{x}_1,\bar{x}_2)$ & 2 & $x_{34}+x_{2}+2x_{24}+2x_{23}+x_{1}\mod 3$\\
106 & 1110011100010000 & 2 & $(0,0,\bar{x}_1,\bar{x}_2)$ & 2 & $x_{34}+x_{2}+2x_{24}+2x_{23}+x_{1}+x_{13}\mod 3$\\
107 & 1110011100011000 & 2 & $(0,0,\bar{x}_1,\bar{x}_2)$ & 2 & $x_{34}+x_{2}+x_{24}+x_{23}+x_{1}\mod 2$\\
108 & 1110100000000000 & 2 & $(0,0,\bar{x}_1,\bar{x}_2)$ & 2 & $x_{34}+x_{24}+2x_{23}+x_{1}+x_{14}+x_{13}+x_{12}\mod 3$\\
109 & 1110100000000001 & 2 & $(0,0,\bar{x}_1,\bar{x}_2)$ & 2 & $x_{34}+x_{24}+2x_{23}+2x_{1}+2x_{14}+2x_{13}+2x_{12}\mod 3$\\
110 & 1110100000010000 & 2 & $(0,0,\bar{x}_1,\bar{x}_2)$ & 2 & $x_{34}+x_{24}+2x_{23}+x_{1}+x_{14}+x_{12}\mod 3$\\
111 & 1110100000010001 & 2 & $(0,0,\bar{x}_1,\bar{x}_2)$ & 2 & $x_{34}+x_{24}+2x_{23}+x_{1}+x_{13}\mod 3$\\
112 & 1110100000010100 & 2 & $(0,0,\bar{x}_1,\bar{x}_2)$ & 2 & $x_{34}+x_{24}+2x_{23}+x_{1}+x_{13}+x_{12}\mod 3$\\
113 & 1110100000010101 & 2 & $(0,0,\bar{x}_1,\bar{x}_2)$ & 2 & $x_{34}+x_{24}+2x_{23}+2x_{1}\mod 3$\\
114 & 1110100000010110 & 2 & $(0,0,\bar{x}_1,\bar{x}_2)$ & 2 & $x_{34}+x_{24}+x_{23}+2x_{1}+x_{14}+x_{13}+x_{12}\mod 5$\\
115 & 1110100000010111 & 2 & $(0,0,\bar{x}_1,\bar{x}_2)$ & 2 & $x_{34}+x_{24}+x_{23}+x_{1}\mod 2$\\
116 & 1110100001000000 & 2 & $(0,0,\bar{x}_1,\bar{x}_2)$ & 2 & $x_{34}+x_{24}+2x_{23}+x_{1}+2x_{14}+x_{12}\mod 3$\\
117 & 1110100001000001 & 2 & $(0,0,\bar{x}_1,\bar{x}_2)$ & 2 & $x_{34}+x_{24}+2x_{23}+x_{1}+2x_{14}+x_{13}+x_{12}\mod 3$\\
118 & 1110100001000010 & 2 & $(0,0,\bar{x}_1,\bar{x}_2)$ & 2 & $x_{34}+x_{24}+x_{23}+x_{1}+x_{14}\mod 2$\\
119 & 1110100001000011 & 2 & $(0,0,\bar{x}_1,\bar{x}_2)$ & 2 & $x_{34}+x_{24}+x_{23}+2x_{1}+3x_{14}+2x_{12}\mod 5$\\
120 & 1110100001010000 & 2 & $(0,0,\bar{x}_1,\bar{x}_2)$ & 2 & $x_{34}+2x_{24}+x_{23}+2x_{1}+x_{14}+2x_{13}+2x_{12}\mod 3$\\
121 & 1110100001010001 & 2 & $(0,0,\bar{x}_1,\bar{x}_2)$ & 2 & $x_{34}+2x_{24}+x_{23}+2x_{1}+x_{14}+2x_{13}\mod 3$\\
122 & 1110100001010010 & 2 & $(0,0,\bar{x}_1,\bar{x}_2)$ & 2 & $x_{34}+2x_{24}+2x_{23}+2x_{1}+x_{14}+2x_{13}\mod 3$\\
123 & 1110100001010011 & 2 & $(0,0,\bar{x}_1,\bar{x}_2)$ & 2 & $x_{34}+x_{24}+2x_{23}+2x_{1}+x_{14}+2x_{13}\mod 3$\\
124 & 1110100001010100 & 2 & $(0,0,\bar{x}_1,\bar{x}_2)$ & 2 & $x_{34}+x_{24}+2x_{23}+2x_{1}+x_{14}+2x_{13}+2x_{12}\mod 3$\\
125 & 1110100001010110 & 2 & $(0,0,\bar{x}_1,\bar{x}_2)$ & 2 & $x_{34}+x_{24}+4x_{23}+3x_{1}+2x_{14}+4x_{13}+4x_{12}\mod 5$\\
126 & 1110100001100000 & 2 & $(0,0,\bar{x}_1,\bar{x}_2)$ & 2 & $x_{34}+x_{24}+2x_{23}+2x_{1}+x_{14}+x_{13}\mod 3$\\
127 & 1110100001100001 & 2 & $(0,0,\bar{x}_1,\bar{x}_2)$ & 2 & $x_{34}+2x_{24}+2x_{23}+2x_{1}+x_{14}+x_{13}\mod 3$\\
128 & 1110100001100100 & 2 & $(0,0,\bar{x}_1,\bar{x}_2)$ & 2 & $x_{34}+x_{24}+2x_{23}+2x_{1}+x_{14}+x_{13}+2x_{12}\mod 3$\\
129 & 1110100001100101 & 2 & $(0,0,\bar{x}_1,\bar{x}_2)$ & 2 & $x_{34}+x_{24}+2x_{23}+3x_{1}+2x_{14}+2x_{13}+4x_{12}\mod 5$\\
130 & 1110100001100110 & 2 & $(0,0,\bar{x}_1,\bar{x}_2)$ & 2 & $x_{34}+x_{24}+x_{23}+3x_{1}+2x_{14}+2x_{13}+4x_{12}\mod 5$\\
131 & 1110100001100111 & 2 & $(0,0,\bar{x}_1,\bar{x}_2)$ & 2 & $x_{34}+x_{24}+x_{23}+2x_{1}+3x_{14}+3x_{13}+4x_{12}\mod 5$\\
132 & 1110100001101000 & 2 & $(0,0,\bar{x}_1,\bar{x}_2)$ & 2 & $x_{34}+x_{24}+x_{23}+2x_{1}+3x_{14}+3x_{13}+3x_{12}\mod 5$\\
133 & 1110100001101001 & 2 & $(0,0,\bar{x}_1,\bar{x}_2)$ & 2 & $x_{34}+x_{24}+4x_{23}+3x_{1}+2x_{14}+2x_{13}+2x_{12}\mod 5$\\
134 & 1110100001110000 & 2 & $(0,0,\bar{x}_1,\bar{x}_2)$ & 2 & $x_{34}+2x_{24}+2x_{23}+x_{1}+2x_{14}+2x_{13}\mod 3$\\
135 & 1110100001110001 & 2 & $(0,0,\bar{x}_1,\bar{x}_2)$ & 2 & $x_{34}+x_{24}+x_{23}+x_{1}+x_{14}+x_{13}\mod 2$\\
136 & 1110100001110100 & 2 & $(0,0,\bar{x}_1,\bar{x}_2)$ & 2 & $x_{34}+2x_{24}+x_{23}+x_{1}+2x_{14}+2x_{13}+x_{12}\mod 3$\\
137 & 1110100001110110 & 2 & $(0,0,\bar{x}_1,\bar{x}_2)$ & 2 & $x_{34}+2x_{24}+2x_{23}+x_{1}+2x_{14}+2x_{13}+x_{12}\mod 3$\\
138 & 1110100001111000 & 2 & $(0,0,\bar{x}_1,\bar{x}_2)$ & 2 & $x_{34}+2x_{24}+3x_{23}+x_{1}+4x_{14}+4x_{13}+4x_{12}\mod 5$\\
139 & 1110100001111100 & 2 & $(0,0,\bar{x}_1,\bar{x}_2)$ & 2 & $x_{34}+x_{24}+2x_{23}+x_{1}+2x_{14}+2x_{13}+2x_{12}\mod 3$\\
140 & 1110100001111110 & 2 & $(0,0,\bar{x}_1,\bar{x}_2)$ & 2 & $x_{34}+x_{24}+x_{23}+x_{1}+x_{14}+x_{13}+x_{12}\mod 2$\\
141 & 1110100010000000 & 2 & $(0,0,\bar{x}_1,\bar{x}_2)$ & 2 & $x_{34}+x_{24}+2x_{23}+2x_{14}+x_{13}+x_{12}\mod 3$\\
142 & 1110100010000001 & 2 & $(0,0,\bar{x}_1,\bar{x}_2)$ & 2 & $x_{34}+x_{24}+x_{23}+x_{14}+x_{13}+x_{12}\mod 2$\\
143 & 1110100100000000 & 2 & $(0,0,\bar{x}_1,\bar{x}_2)$ & 2 & $x_{34}+x_{24}+x_{23}+x_{1}\mod 3$\\
144 & 1110100100000001 & 2 & $(0,0,\bar{x}_1,\bar{x}_2)$ & 2 & $x_{34}+x_{24}+x_{23}+2x_{1}+2x_{13}+2x_{12}\mod 3$\\
145 & 1110100100010000 & 2 & $(0,0,\bar{x}_1,\bar{x}_2)$ & 2 & $x_{34}+x_{24}+x_{23}+2x_{1}+2x_{12}\mod 3$\\
146 & 1110100100010001 & 2 & $(0,0,\bar{x}_1,\bar{x}_2)$ & 2 & $x_{34}+x_{24}+3x_{23}+x_{1}+3x_{13}+x_{12}\mod 5$\\
147 & 1110100100010100 & 2 & $(0,0,\bar{x}_1,\bar{x}_2)$ & 2 & $x_{34}+x_{24}+x_{23}+x_{1}+x_{14}\mod 3$\\
148 & 1110100100010101 & 2 & $(0,0,\bar{x}_1,\bar{x}_2)$ & 2 & $x_{34}+x_{24}+x_{23}+x_{1}+x_{13}+x_{12}\mod 3$\\
149 & 1110100100010110 & 2 & $(0,0,\bar{x}_1,\bar{x}_2)$ & 2 & $x_{34}+x_{24}+x_{23}+2x_{1}\mod 3$\\
150 & 1110100101000000 & 2 & $(0,0,\bar{x}_1,\bar{x}_2)$ & 2 & $x_{34}+x_{24}+x_{23}+x_{1}+2x_{14}+x_{13}+x_{12}\mod 3$\\
151 & 1110100101000001 & 2 & $(0,0,\bar{x}_1,\bar{x}_2)$ & 2 & $x_{34}+x_{24}+x_{23}+x_{1}+2x_{14}\mod 3$\\
152 & 1110100101000010 & 2 & $(0,0,\bar{x}_1,\bar{x}_2)$ & 2 & $x_{34}+x_{24}+x_{23}+x_{1}+2x_{14}+x_{12}\mod 3$\\
153 & 1110100101000011 & 2 & $(0,0,\bar{x}_1,\bar{x}_2)$ & 2 & $x_{34}+x_{24}+x_{23}+2x_{1}+x_{14}\mod 3$\\
154 & 1110100101010000 & 2 & $(0,0,\bar{x}_1,\bar{x}_2)$ & 2 & $x_{34}+x_{24}+x_{23}+2x_{1}+x_{14}+2x_{13}\mod 3$\\
155 & 1110100101010001 & 2 & $(0,0,\bar{x}_1,\bar{x}_2)$ & 2 & $x_{34}+x_{24}+3x_{23}+3x_{1}+2x_{14}+4x_{13}+x_{12}\mod 5$\\
156 & 1110100101010010 & 2 & $(0,0,\bar{x}_1,\bar{x}_2)$ & 2 & $x_{34}+x_{24}+3x_{23}+3x_{1}+2x_{14}+4x_{13}\mod 5$\\
157 & 1110100101010011 & 2 & $(0,0,\bar{x}_1,\bar{x}_2)$ & 2 & $x_{34}+x_{24}+3x_{23}+2x_{1}+3x_{14}+4x_{13}+x_{12}\mod 5$\\
158 & 1110100101010100 & 2 & $(0,0,\bar{x}_1,\bar{x}_2)$ & 2 & $x_{34}+x_{24}+x_{23}+2x_{1}+x_{14}+2x_{13}+2x_{12}\mod 3$\\
159 & 1110100101010110 & 2 & $(0,0,\bar{x}_1,\bar{x}_2)$ & 2 & $x_{34}+x_{24}+3x_{23}+4x_{1}+x_{14}+4x_{13}+4x_{12}\mod 5$\\
160 & 1110100101100000 & 2 & $(0,0,\bar{x}_1,\bar{x}_2)$ & 2 & $x_{34}+x_{24}+x_{23}+2x_{1}+x_{14}+x_{13}\mod 3$\\
161 & 1110100101100001 & 2 & $(0,0,\bar{x}_1,\bar{x}_2)$ & 2 & $x_{34}+x_{24}+3x_{23}+3x_{1}+2x_{14}+2x_{13}+3x_{12}\mod 5$\\
162 & 1110100101100100 & 2 & $(0,0,\bar{x}_1,\bar{x}_2)$ & 2 & $x_{34}+x_{24}+3x_{23}+2x_{1}+3x_{14}+3x_{13}+4x_{12}\mod 5$\\
163 & 1110100101100101 & 2 & $(0,0,\bar{x}_1,\bar{x}_2)$ & 2 & $x_{34}+x_{24}+3x_{23}+4x_{1}+x_{14}+x_{13}+4x_{12}\mod 5$\\
164 & 1110100101100110 & 2 & $(0,0,\bar{x}_1,\bar{x}_2)$ & 2 & $2x_{34}+4x_{24}+4x_{23}+3x_{1}+2x_{14}+2x_{13}+x_{12}\mod 5$\\
165 & 1110100101100111 & 2 & $(0,0,\bar{x}_1,\bar{x}_2)$ & 2 & $2x_{34}+2x_{24}+2x_{23}+x_{1}+2x_{14}+2x_{13}+x_{12}\mod 3$\\
166 & 1110100101101000 & 2 & $(0,0,\bar{x}_1,\bar{x}_2)$ & 2 & $x_{34}+x_{24}+x_{23}+2x_{1}+x_{14}+x_{13}+x_{12}\mod 3$\\
167 & 1110100101101001 & 2 & $(0,0,\bar{x}_1,\bar{x}_2)$ & 2 & $6x_{34}+6x_{24}+18x_{23}+15x_{1}+15x_{14}+15x_{13}+15x_{12}\mod 30$\\
168 & 1110100101110000 & 2 & $(0,0,\bar{x}_1,\bar{x}_2)$ & 2 & $x_{34}+x_{24}+x_{23}+x_{1}+2x_{14}+2x_{13}\mod 3$\\
169 & 1110100101110001 & 2 & $(0,0,\bar{x}_1,\bar{x}_2)$ & 2 & $x_{34}+x_{24}+x_{23}+x_{1}+2x_{14}+2x_{13}+x_{12}\mod 3$\\
170 & 1110100101110100 & 2 & $(0,0,\bar{x}_1,\bar{x}_2)$ & 2 & $x_{34}+3x_{24}+x_{23}+x_{1}+4x_{14}+4x_{13}+2x_{12}\mod 5$\\
171 & 1110100101110110 & 2 & $(0,0,\bar{x}_1,\bar{x}_2)$ & 2 & $2x_{34}+4x_{24}+4x_{23}+2x_{1}+3x_{14}+3x_{13}+x_{12}\mod 5$\\
172 & 1110100101111000 & 2 & $(0,0,\bar{x}_1,\bar{x}_2)$ & 2 & $x_{34}+2x_{24}+2x_{23}+x_{1}+4x_{14}+4x_{13}+4x_{12}\mod 5$\\
173 & 1110100101111001 & 2 & $(0,0,\bar{x}_1,\bar{x}_2)$ & 2 & $6x_{34}+x_{24}+5x_{23}+6x_{1}+6x_{14}+6x_{13}+6x_{12}\mod 12$\\
174 & 1110100101111100 & 2 & $(0,0,\bar{x}_1,\bar{x}_2)$ & 2 & $x_{34}+x_{24}+3x_{23}+x_{1}+4x_{14}+4x_{13}+4x_{12}\mod 5$\\
175 & 1110100101111110 & 2 & $(0,0,\bar{x}_1,\bar{x}_2)$ & 2 & $4x_{34}+4x_{24}+4x_{23}+4x_{1}+2x_{14}+2x_{13}+2x_{12}\mod 6$\\
176 & 1110100110000000 & 2 & $(0,0,\bar{x}_1,\bar{x}_2)$ & 2 & $x_{34}+x_{24}+x_{23}+x_{14}+2x_{13}+2x_{12}\mod 3$\\
177 & 1110100110000001 & 2 & $(0,0,\bar{x}_1,\bar{x}_2)$ & 2 & $x_{34}+x_{24}+x_{23}+2x_{14}+2x_{13}+2x_{12}\mod 3$\\
178 & 1110100110010000 & 2 & $(0,0,\bar{x}_1,\bar{x}_2)$ & 2 & $x_{34}+x_{24}+x_{23}+x_{14}+x_{13}+2x_{12}\mod 3$\\
179 & 1110100110010001 & 2 & $(0,0,\bar{x}_1,\bar{x}_2)$ & 2 & $x_{34}+x_{24}+3x_{23}+x_{14}+3x_{13}+x_{12}\mod 5$\\
180 & 1110100110010100 & 2 & $(0,0,\bar{x}_1,\bar{x}_2)$ & 2 & $x_{34}+x_{24}+3x_{23}+x_{14}+3x_{13}+3x_{12}\mod 5$\\
181 & 1110100110010101 & 2 & $(0,0,\bar{x}_1,\bar{x}_2)$ & 3 & $x_{34}+x_{24}+x_{23}+x_{14}+x_{13}+2x_{12}+2x_{124}\mod 3$\\
182 & 1110100110010110 & 2 & $(0,0,\bar{x}_1,\bar{x}_2)$ & 2 & $x_{34}+x_{24}+28x_{23}+13x_{14}+16x_{13}+16x_{12}\mod 30$\\
183 & 1110100110010111 & 2 & $(0,0,\bar{x}_1,\bar{x}_2)$ & 2 & $x_{34}+x_{24}+x_{23}+x_{14}+x_{13}+x_{12}\mod 3$\\
184 & 1111000000000000 & 1 & $(0,\bar{x}_1,0,0)$ & 1 & $x_{2}+x_{1}\mod 3$\\
185 & 1111000000001000 & 2 & $(\bar{x}_1,\bar{x}_1,0,\bar{x}_2)$ & 2 & $x_{2}+x_{23}+x_{1}+x_{14}+x_{12}\mod 3$\\
186 & 1111000000001001 & 2 & $(\bar{x}_1,\bar{x}_1,0,\bar{x}_2)$ & 2 & $x_{2}+x_{23}+2x_{1}+2x_{14}\mod 3$\\
187 & 1111000000001100 & 2 & $(\bar{x}_1,\bar{x}_1,\bar{x}_2,0)$ & 2 & $x_{2}+x_{1}+x_{13}+x_{12}\mod 3$\\
188 & 1111000000001110 & 3 & $(\bar{x}_1,\bar{x}_1,\bar{x}_2,\bar{x}_3)$ & 3 & $x_{2}+x_{1}+x_{134}+x_{12}\mod 3$\\
189 & 1111000000001111 & 1 & $(0,\bar{x}_1,0,0)$ & 1 & $x_{2}+x_{1}\mod 2$\\
190 & 1111100000000000 & 2 & $(0,\bar{x}_1,0,\bar{x}_2)$ & 2 & $x_{24}+x_{23}+2x_{1}+2x_{13}+2x_{12}\mod 3$\\
191 & 1111100000000001 & 2 & $(0,\bar{x}_1,0,\bar{x}_2)$ & 2 & $x_{24}+x_{23}+x_{1}\mod 3$\\
192 & 1111100000000100 & 2 & $(0,\bar{x}_1,0,\bar{x}_2)$ & 2 & $x_{24}+x_{23}+x_{1}+x_{13}+x_{12}\mod 3$\\
193 & 1111100000000101 & 2 & $(0,\bar{x}_1,0,\bar{x}_2)$ & 2 & $x_{24}+x_{23}+2x_{1}+2x_{13}\mod 3$\\
194 & 1111100000000110 & 2 & $(0,\bar{x}_1,0,\bar{x}_2)$ & 2 & $x_{24}+x_{23}+x_{1}+x_{12}\mod 3$\\
195 & 1111100000000111 & 3 & $(\bar{x}_1,\bar{x}_1,\bar{x}_2,\bar{x}_3)$ & 3 & $x_{24}+x_{23}+x_{234}+x_{1}\mod 2$\\
196 & 1111100000001000 & 2 & $(0,\bar{x}_1,0,\bar{x}_2)$ & 2 & $x_{24}+x_{23}+x_{1}+2x_{12}\mod 3$\\
197 & 1111100000001001 & 2 & $(0,\bar{x}_1,0,\bar{x}_2)$ & 2 & $x_{24}+x_{23}+x_{1}+x_{13}+2x_{12}\mod 3$\\
198 & 1111100000001100 & 2 & $(0,\bar{x}_1,0,\bar{x}_2)$ & 2 & $x_{24}+x_{23}+2x_{1}+2x_{14}+x_{12}\mod 3$\\
199 & 1111100000001101 & 3 & $(\bar{x}_1,\bar{x}_1,\bar{x}_2,\bar{x}_3)$ & 3 & $x_{24}+x_{23}+x_{234}+x_{1}+x_{12}+x_{124}\mod 2$\\
200 & 1111100000001110 & 3 & $(\bar{x}_1,\bar{x}_1,\bar{x}_2,\bar{x}_3)$ & 3 & $x_{24}+x_{23}+x_{234}+x_{1}+x_{12}+x_{124}+x_{123}\mod 2$\\
201 & 1111100000001111 & 2 & $(0,\bar{x}_1,0,\bar{x}_2)$ & 2 & $x_{24}+x_{23}+3x_{1}+4x_{14}+4x_{13}+2x_{12}\mod 5$\\
202 & 1111100000010000 & 2 & $(0,\bar{x}_1,0,\bar{x}_2)$ & 2 & $x_{24}+x_{23}+2x_{1}+2x_{14}+2x_{13}\mod 3$\\
203 & 1111100000010100 & 2 & $(0,\bar{x}_1,0,\bar{x}_2)$ & 2 & $x_{24}+x_{23}+x_{1}+x_{14}+3x_{13}+2x_{12}\mod 5$\\
204 & 1111100000010110 & 2 & $(0,\bar{x}_1,0,\bar{x}_2)$ & 2 & $x_{24}+x_{23}+x_{1}+x_{14}+x_{13}\mod 3$\\
205 & 1111100000011000 & 2 & $(0,\bar{x}_1,0,\bar{x}_2)$ & 2 & $x_{24}+x_{23}+x_{1}+x_{14}+x_{13}+2x_{12}\mod 3$\\
206 & 1111100000011001 & 2 & $(0,\bar{x}_1,0,\bar{x}_2)$ & 2 & $x_{24}+x_{23}+2x_{1}+x_{14}+2x_{13}+3x_{12}\mod 5$\\
207 & 1111100000011100 & 2 & $(0,\bar{x}_1,0,\bar{x}_2)$ & 2 & $x_{24}+x_{23}+3x_{1}+4x_{14}+3x_{13}+2x_{12}\mod 5$\\
208 & 1111100000011101 & 3 & $(\bar{x}_1,\bar{x}_1,\bar{x}_2,\bar{x}_3)$ & 3 & $x_{24}+x_{23}+x_{1}+2x_{134}+2x_{12}+2x_{124}\mod 3$\\
209 & 1111100000011110 & 3 & $(\bar{x}_1,\bar{x}_1,\bar{x}_2,\bar{x}_3)$ & 3 & $x_{24}+x_{23}+x_{1}+2x_{134}+2x_{12}+2x_{124}+2x_{123}\mod 3$\\
210 & 1111100000011111 & 2 & $(0,\bar{x}_1,0,\bar{x}_2)$ & 2 & $x_{24}+x_{23}+2x_{1}+2x_{14}+2x_{13}+x_{12}\mod 3$\\
211 & 1111100001000000 & 2 & $(0,\bar{x}_1,0,\bar{x}_2)$ & 2 & $x_{24}+x_{23}+x_{1}+2x_{14}+x_{13}+x_{12}\mod 3$\\
212 & 1111100001000010 & 2 & $(0,\bar{x}_1,0,\bar{x}_2)$ & 2 & $x_{24}+x_{23}+2x_{1}+3x_{14}+x_{13}+x_{12}\mod 5$\\
213 & 1111100001000011 & 2 & $(0,\bar{x}_1,0,\bar{x}_2)$ & 2 & $x_{24}+x_{23}+x_{1}+2x_{14}+x_{13}\mod 3$\\
214 & 1111100001001000 & 2 & $(0,\bar{x}_1,0,\bar{x}_2)$ & 2 & $x_{24}+x_{23}+2x_{1}+3x_{14}+x_{13}+3x_{12}\mod 5$\\
215 & 1111100001001001 & 2 & $(0,\bar{x}_1,0,\bar{x}_2)$ & 2 & $x_{24}+x_{23}+3x_{1}+2x_{14}+x_{13}+2x_{12}\mod 5$\\
216 & 1111100001001010 & 2 & $(0,\bar{x}_1,0,\bar{x}_2)$ & 2 & $x_{24}+x_{23}+2x_{1}+x_{14}+2x_{13}+x_{12}\mod 3$\\
217 & 1111100001001011 & 3 & $(\bar{x}_1,\bar{x}_1,\bar{x}_2,\bar{x}_3)$ & 3 & $x_{24}+x_{23}+x_{1}+2x_{14}+x_{134}+2x_{12}+2x_{124}+2x_{123}\mod 3$\\
218 & 1111100001001100 & 2 & $(0,\bar{x}_1,0,\bar{x}_2)$ & 2 & $x_{24}+x_{23}+x_{1}+2x_{14}+x_{13}+2x_{12}\mod 3$\\
219 & 1111100001001110 & 3 & $(\bar{x}_1,\bar{x}_1,\bar{x}_2,\bar{x}_3)$ & 3 & $x_{24}+x_{23}+x_{1}+2x_{14}+x_{134}+2x_{12}+2x_{123}\mod 3$\\
220 & 1111100001001111 & 2 & $(0,\bar{x}_1,0,\bar{x}_2)$ & 2 & $x_{24}+2x_{23}+x_{1}+4x_{14}+3x_{13}+4x_{12}\mod 5$\\
221 & 1111100010000000 & 2 & $(0,\bar{x}_1,0,\bar{x}_2)$ & 2 & $x_{24}+x_{23}+2x_{14}+2x_{13}+x_{12}\mod 3$\\
222 & 1111100010000001 & 2 & $(0,\bar{x}_1,0,\bar{x}_2)$ & 2 & $x_{24}+x_{23}+x_{14}+x_{13}+2x_{12}\mod 3$\\
223 & 1111100010000100 & 2 & $(0,\bar{x}_1,0,\bar{x}_2)$ & 2 & $x_{24}+x_{23}+x_{14}+2x_{13}+3x_{12}\mod 5$\\
224 & 1111100010000101 & 2 & $(0,\bar{x}_1,0,\bar{x}_2)$ & 2 & $x_{24}+x_{23}+2x_{14}+4x_{13}+2x_{12}\mod 5$\\
225 & 1111100010000110 & 2 & $(0,\bar{x}_1,0,\bar{x}_2)$ & 2 & $x_{24}+x_{23}+x_{14}+x_{13}+x_{12}\mod 3$\\
226 & 1111100010000111 & 3 & $(\bar{x}_1,\bar{x}_1,\bar{x}_2,\bar{x}_3)$ & 3 & $x_{24}+x_{23}+x_{14}+x_{13}+2x_{134}+2x_{12}+2x_{124}+2x_{123}\mod 3$\\
227 & 1111100010001000 & 2 & $(0,\bar{x}_1,0,\bar{x}_2)$ & 2 & $x_{24}+x_{23}+x_{14}+x_{13}\mod 3$\\
228 & 1111100010001001 & 2 & $(0,\bar{x}_1,0,\bar{x}_2)$ & 2 & $x_{24}+x_{23}+x_{14}+2x_{13}\mod 5$\\
229 & 1111100100000000 & 2 & $(0,\bar{x}_1,0,\bar{x}_2)$ & 2 & $x_{24}+2x_{23}+x_{1}+x_{13}\mod 3$\\
230 & 1111100100000100 & 2 & $(0,\bar{x}_1,0,\bar{x}_2)$ & 2 & $x_{24}+2x_{23}+x_{1}+x_{12}\mod 3$\\
231 & 1111100100000110 & 2 & $(0,\bar{x}_1,0,\bar{x}_2)$ & 2 & $x_{24}+x_{23}+x_{1}\mod 2$\\
232 & 1111100100001000 & 2 & $(0,\bar{x}_1,0,\bar{x}_2)$ & 2 & $x_{24}+2x_{23}+x_{1}+x_{14}+2x_{12}\mod 3$\\
233 & 1111100100001001 & 2 & $(0,\bar{x}_1,0,\bar{x}_2)$ & 2 & $x_{24}+x_{23}+x_{1}+x_{12}\mod 2$\\
234 & 1111100100001100 & 2 & $(0,\bar{x}_1,0,\bar{x}_2)$ & 2 & $x_{24}+2x_{23}+2x_{1}+2x_{14}+x_{12}\mod 3$\\
235 & 1111100100001101 & 3 & $(\bar{x}_1,\bar{x}_1,\bar{x}_2,\bar{x}_3)$ & 3 & $x_{24}+x_{23}+x_{234}+x_{1}+2x_{12}+2x_{124}+x_{123}\mod 3$\\
236 & 1111100100001110 & 3 & $(\bar{x}_1,\bar{x}_1,\bar{x}_2,\bar{x}_3)$ & 3 & $x_{24}+x_{23}+x_{234}+x_{1}+2x_{12}+2x_{124}+2x_{123}\mod 3$\\
237 & 1111100100001111 & 2 & $(0,\bar{x}_1,0,\bar{x}_2)$ & 2 & $x_{24}+4x_{23}+2x_{1}+4x_{14}+x_{13}+3x_{12}\mod 5$\\
238 & 1111100101000000 & 2 & $(0,\bar{x}_1,0,\bar{x}_2)$ & 2 & $x_{24}+2x_{23}+x_{1}+2x_{14}+x_{13}\mod 3$\\
239 & 1111100101000010 & 2 & $(0,\bar{x}_1,0,\bar{x}_2)$ & 2 & $x_{24}+2x_{23}+2x_{1}+x_{14}+2x_{13}\mod 3$\\
240 & 1111100101000100 & 2 & $(0,\bar{x}_1,0,\bar{x}_2)$ & 2 & $x_{24}+2x_{23}+2x_{1}+x_{14}+2x_{13}+2x_{12}\mod 3$\\
241 & 1111100101000110 & 2 & $(0,\bar{x}_1,0,\bar{x}_2)$ & 2 & $x_{24}+4x_{23}+3x_{1}+2x_{14}+4x_{13}+4x_{12}\mod 5$\\
242 & 1111100101001000 & 2 & $(0,\bar{x}_1,0,\bar{x}_2)$ & 2 & $x_{24}+4x_{23}+2x_{1}+3x_{14}+4x_{13}+3x_{12}\mod 5$\\
243 & 1111100101001001 & 2 & $(0,\bar{x}_1,0,\bar{x}_2)$ & 2 & $x_{24}+2x_{23}+2x_{1}+x_{14}+2x_{13}+x_{12}\mod 3$\\
244 & 1111100101001010 & 2 & $(0,\bar{x}_1,0,\bar{x}_2)$ & 2 & $x_{24}+4x_{23}+2x_{1}+3x_{14}+x_{13}+3x_{12}\mod 5$\\
245 & 1111100101001011 & 3 & $(\bar{x}_1,\bar{x}_1,\bar{x}_2,\bar{x}_3)$ & 3 & $x_{24}+x_{23}+x_{1}+x_{14}+x_{134}+x_{12}+x_{124}+x_{123}\mod 2$\\
246 & 1111100101001100 & 2 & $(0,\bar{x}_1,0,\bar{x}_2)$ & 2 & $x_{24}+4x_{23}+x_{1}+4x_{14}+2x_{13}+4x_{12}\mod 5$\\
247 & 1111100101001101 & 3 & $(\bar{x}_1,\bar{x}_1,\bar{x}_2,\bar{x}_3)$ & 3 & $x_{24}+x_{23}+x_{1}+x_{14}+x_{134}+x_{12}\mod 2$\\
248 & 1111100101001110 & 3 & $(\bar{x}_1,\bar{x}_1,\bar{x}_2,\bar{x}_3)$ & 3 & $x_{24}+x_{23}+x_{1}+x_{14}+x_{134}+x_{12}+x_{123}\mod 2$\\
249 & 1111100101001111 & 2 & $(0,\bar{x}_1,0,\bar{x}_2)$ & 2 & $x_{24}+2x_{23}+x_{1}+2x_{14}+x_{13}+2x_{12}\mod 3$\\
250 & 1111100101100000 & 2 & $(0,\bar{x}_1,0,\bar{x}_2)$ & 2 & $x_{24}+x_{23}+x_{1}+x_{14}+x_{13}\mod 2$\\
251 & 1111100101101000 & 2 & $(0,\bar{x}_1,0,\bar{x}_2)$ & 2 & $x_{24}+4x_{23}+2x_{1}+3x_{14}+3x_{13}+3x_{12}\mod 5$\\
252 & 1111100101101001 & 2 & $(0,\bar{x}_1,0,\bar{x}_2)$ & 2 & $4x_{24}+2x_{23}+3x_{1}+3x_{14}+3x_{13}+3x_{12}\mod 6$\\
253 & 1111100101101100 & 2 & $(0,\bar{x}_1,0,\bar{x}_2)$ & 2 & $x_{24}+2x_{23}+x_{1}+2x_{14}+2x_{13}+2x_{12}\mod 3$\\
254 & 1111100101101110 & 3 & $(\bar{x}_1,\bar{x}_1,\bar{x}_2,\bar{x}_3)$ & 3 & $x_{24}+x_{23}+x_{234}+x_{1}+2x_{14}+2x_{13}+2x_{12}\mod 3$\\
255 & 1111100101101111 & 2 & $(0,\bar{x}_1,0,\bar{x}_2)$ & 2 & $3x_{24}+3x_{23}+3x_{1}+3x_{14}+3x_{13}+3x_{12}\mod 6$\\
256 & 1111100110000000 & 2 & $(0,\bar{x}_1,0,\bar{x}_2)$ & 2 & $x_{24}+2x_{23}+x_{14}+x_{13}+2x_{12}\mod 3$\\
257 & 1111100110000001 & 2 & $(0,\bar{x}_1,0,\bar{x}_2)$ & 2 & $x_{24}+4x_{23}+x_{14}+2x_{13}+2x_{12}\mod 5$\\
258 & 1111100110000100 & 2 & $(0,\bar{x}_1,0,\bar{x}_2)$ & 2 & $x_{24}+4x_{23}+x_{14}+2x_{13}+3x_{12}\mod 5$\\
259 & 1111100110000101 & 2 & $(0,\bar{x}_1,0,\bar{x}_2)$ & 2 & $x_{24}+2x_{23}+x_{14}+x_{13}+x_{12}\mod 3$\\
260 & 1111100110000110 & 2 & $(0,\bar{x}_1,0,\bar{x}_2)$ & 2 & $x_{24}+4x_{23}+x_{14}+3x_{13}+3x_{12}\mod 5$\\
261 & 1111100110000111 & 3 & $(\bar{x}_1,\bar{x}_1,\bar{x}_2,\bar{x}_3)$ & 3 & $x_{24}+x_{23}+x_{14}+x_{13}+x_{134}+x_{12}+x_{124}+x_{123}\mod 2$\\
262 & 1111100110001000 & 2 & $(0,\bar{x}_1,0,\bar{x}_2)$ & 2 & $x_{24}+4x_{23}+x_{14}+2x_{13}\mod 5$\\
263 & 1111100110001001 & 3 & $(\bar{x}_1,\bar{x}_2,\bar{x}_3,\bar{x}_3)$ & 3 & $x_{24}+x_{23}+x_{234}+x_{14}+x_{13}+2x_{124}+2x_{123}\mod 3$\\
264 & 1111100110001100 & 2 & $(0,\bar{x}_1,0,\bar{x}_2)$ & 2 & $x_{24}+2x_{23}+2x_{14}+2x_{13}\mod 3$\\
265 & 1111100110001101 & 3 & $(\bar{x}_1,\bar{x}_1,\bar{x}_2,\bar{x}_3)$ & 3 & $x_{24}+x_{23}+x_{14}+x_{13}+x_{134}+x_{123}\mod 2$\\
266 & 1111100110010000 & 2 & $(0,\bar{x}_1,0,\bar{x}_2)$ & 2 & $x_{24}+x_{23}+x_{14}+x_{13}+x_{12}\mod 2$\\
267 & 1111100110010100 & 2 & $(0,\bar{x}_1,0,\bar{x}_2)$ & 2 & $x_{24}+2x_{23}+x_{14}+2x_{13}+x_{12}\mod 3$\\
268 & 1111100110010110 & 2 & $(0,\bar{x}_1,0,\bar{x}_2)$ & 2 & $x_{24}+11x_{23}+5x_{14}+7x_{13}+6x_{12}\mod 12$\\
269 & 1111100110011000 & 3 & $(\bar{x}_1,\bar{x}_2,\bar{x}_3,\bar{x}_3)$ & 3 & $x_{24}+x_{23}+x_{234}+x_{14}+x_{13}+x_{134}+2x_{123}\mod 3$\\
270 & 1111100110011001 & 2 & $(0,\bar{x}_1,0,\bar{x}_2)$ & 2 & $x_{24}+2x_{23}+x_{14}+2x_{13}\mod 3$\\
271 & 1111100110011100 & 3 & $(\bar{x}_1,\bar{x}_2,\bar{x}_3,\bar{x}_3)$ & 3 & $x_{24}+x_{23}+x_{14}+x_{13}+x_{123}\mod 2$\\
272 & 1111100110011101 & 3 & $(\bar{x}_1,\bar{x}_1,\bar{x}_2,\bar{x}_3)$ & 3 & $x_{24}+x_{23}+x_{234}+x_{14}+x_{13}+x_{134}+x_{124}+2x_{123}\mod 3$\\
273 & 1111100110011110 & 3 & $(\bar{x}_1,\bar{x}_1,\bar{x}_2,\bar{x}_3)$ & 3 & $x_{24}+x_{23}+x_{234}+x_{14}+x_{13}+x_{134}+x_{124}+x_{123}\mod 3$\\
274 & 1111100110011111 & 2 & $(0,\bar{x}_1,0,\bar{x}_2)$ & 2 & $x_{24}+x_{23}+x_{14}+x_{13}\mod 2$\\
275 & 1111110000000000 & 2 & $(0,\bar{x}_1,\bar{x}_2,0)$ & 2 & $x_{23}+x_{1}\mod 3$\\
276 & 1111110000000010 & 2 & $(0,\bar{x}_1,\bar{x}_2,0)$ & 2 & $x_{23}+2x_{1}+2x_{14}\mod 3$\\
277 & 1111110000000011 & 2 & $(0,\bar{x}_1,\bar{x}_2,0)$ & 2 & $x_{23}+x_{1}\mod 2$\\
278 & 1111110000100000 & 2 & $(0,\bar{x}_1,\bar{x}_2,0)$ & 2 & $x_{23}+x_{1}+x_{14}+2x_{13}\mod 3$\\
279 & 1111110000100001 & 2 & $(0,\bar{x}_1,\bar{x}_2,0)$ & 2 & $x_{23}+2x_{1}+2x_{14}+x_{13}\mod 3$\\
280 & 1111110000100010 & 2 & $(0,\bar{x}_1,\bar{x}_2,0)$ & 2 & $x_{23}+2x_{1}+x_{14}+3x_{13}+4x_{12}\mod 5$\\
281 & 1111110000100011 & 3 & $(\bar{x}_1,\bar{x}_2,\bar{x}_1,\bar{x}_3)$ & 3 & $x_{23}+x_{1}+2x_{13}+x_{134}+x_{12}+2x_{124}+x_{123}\mod 3$\\
282 & 1111110000100100 & 2 & $(0,\bar{x}_1,\bar{x}_2,0)$ & 2 & $x_{23}+2x_{1}+x_{14}+3x_{13}+2x_{12}\mod 5$\\
283 & 1111110000100110 & 2 & $(0,\bar{x}_1,\bar{x}_2,0)$ & 2 & $x_{23}+2x_{1}+2x_{14}+x_{13}+2x_{12}\mod 3$\\
284 & 1111110000100111 & 3 & $(\bar{x}_1,\bar{x}_1,\bar{x}_2,\bar{x}_3)$ & 3 & $x_{23}+x_{1}+x_{13}+x_{134}+x_{124}+x_{123}\mod 2$\\
285 & 1111110000101000 & 2 & $(0,\bar{x}_1,\bar{x}_2,0)$ & 2 & $x_{23}+2x_{1}+2x_{14}+x_{13}+x_{12}\mod 3$\\
286 & 1111110000101001 & 2 & $(0,\bar{x}_1,\bar{x}_2,0)$ & 2 & $x_{23}+2x_{1}+x_{14}+3x_{13}+3x_{12}\mod 5$\\
287 & 1111110000101010 & 2 & $(0,\bar{x}_1,\bar{x}_2,0)$ & 2 & $x_{23}+x_{1}+x_{14}+2x_{13}+2x_{12}\mod 3$\\
288 & 1111110000101011 & 3 & $(\bar{x}_1,\bar{x}_1,\bar{x}_2,\bar{x}_3)$ & 3 & $x_{23}+x_{1}+x_{13}+x_{134}+x_{12}+x_{124}\mod 2$\\
289 & 1111110000110000 & 2 & $(0,\bar{x}_1,\bar{x}_2,0)$ & 2 & $x_{23}+x_{1}+x_{13}\mod 2$\\
290 & 1111110000111000 & 3 & $(\bar{x}_1,\bar{x}_2,\bar{x}_2,\bar{x}_3)$ & 3 & $x_{23}+x_{1}+2x_{13}+2x_{12}+x_{124}+x_{123}\mod 3$\\
291 & 1111110000111001 & 3 & $(\bar{x}_1,\bar{x}_2,\bar{x}_1,\bar{x}_3)$ & 3 & $x_{23}+x_{1}+x_{13}+x_{12}+x_{124}+x_{123}\mod 2$\\
292 & 1111110000111010 & 3 & $(\bar{x}_1,\bar{x}_2,\bar{x}_1,\bar{x}_3)$ & 3 & $x_{23}+x_{1}+x_{13}+x_{12}+x_{124}\mod 2$\\
293 & 1111110000111100 & 2 & $(0,\bar{x}_1,\bar{x}_2,0)$ & 2 & $x_{23}+2x_{1}+x_{13}+x_{12}\mod 3$\\
294 & 1111110000111110 & 3 & $(\bar{x}_1,\bar{x}_1,\bar{x}_2,\bar{x}_3)$ & 3 & $x_{23}+x_{1}+x_{14}+2x_{13}+2x_{134}+2x_{12}+2x_{124}\mod 3$\\
295 & 1111110000111111 & 2 & $(0,\bar{x}_1,\bar{x}_2,0)$ & 2 & $x_{23}+x_{1}+x_{13}+x_{12}\mod 2$\\
296 & 1111110010000000 & 2 & $(0,\bar{x}_1,\bar{x}_2,0)$ & 2 & $x_{23}+2x_{14}+2x_{13}+2x_{12}\mod 3$\\
297 & 1111110010000001 & 2 & $(0,\bar{x}_1,\bar{x}_2,0)$ & 2 & $x_{23}+x_{14}+x_{13}+2x_{12}\mod 5$\\
298 & 1111110010000010 & 2 & $(0,\bar{x}_1,\bar{x}_2,0)$ & 2 & $x_{23}+x_{14}+x_{13}+x_{12}\mod 3$\\
299 & 1111110010000011 & 2 & $(0,\bar{x}_1,\bar{x}_2,0)$ & 3 & $x_{23}+x_{14}+x_{13}+x_{12}+2x_{124}\mod 3$\\
300 & 1111110010010000 & 2 & $(0,\bar{x}_1,\bar{x}_2,0)$ & 2 & $x_{23}+x_{14}+2x_{13}+x_{12}\mod 3$\\
301 & 1111110010010001 & 2 & $(0,\bar{x}_1,\bar{x}_2,0)$ & 2 & $x_{23}+2x_{14}+x_{13}+2x_{12}\mod 3$\\
302 & 1111110010010010 & 2 & $(0,\bar{x}_1,\bar{x}_2,0)$ & 2 & $x_{23}+2x_{14}+3x_{13}+x_{12}\mod 5$\\
303 & 1111110010010011 & 3 & $(\bar{x}_1,\bar{x}_2,\bar{x}_1,\bar{x}_3)$ & 3 & $x_{23}+x_{14}+x_{13}+x_{12}+x_{124}+x_{123}\mod 2$\\
304 & 1111110010010100 & 2 & $(0,\bar{x}_1,\bar{x}_2,0)$ & 2 & $x_{23}+2x_{14}+3x_{13}+3x_{12}\mod 5$\\
305 & 1111110010010101 & 2 & $(0,\bar{x}_1,\bar{x}_2,0)$ & 2 & $x_{23}+x_{14}+x_{13}+x_{12}\mod 2$\\
306 & 1111110010010110 & 2 & $(0,\bar{x}_1,\bar{x}_2,0)$ & 2 & $2x_{23}+x_{14}+2x_{13}+2x_{12}\mod 3$\\
307 & 1111110010010111 & 3 & $(\bar{x}_1,\bar{x}_1,\bar{x}_2,\bar{x}_3)$ & 3 & $x_{23}+x_{14}+x_{13}+x_{134}+x_{12}+x_{124}\mod 3$\\
308 & 1111110010100000 & 2 & $(0,\bar{x}_1,\bar{x}_2,0)$ & 2 & $x_{23}+x_{14}+x_{12}\mod 5$\\
309 & 1111110010100001 & 2 & $(0,\bar{x}_1,\bar{x}_2,0)$ & 2 & $x_{23}+x_{14}+x_{12}\mod 3$\\
310 & 1111110010100100 & 2 & $(0,\bar{x}_1,\bar{x}_2,0)$ & 2 & $x_{23}+2x_{14}+x_{12}\mod 3$\\
311 & 1111110010100101 & 3 & $(\bar{x}_1,\bar{x}_2,\bar{x}_3,\bar{x}_2)$ & 3 & $x_{23}+x_{14}+x_{12}+x_{123}\mod 2$\\
312 & 1111110010100110 & 2 & $(0,\bar{x}_1,\bar{x}_2,0)$ & 2 & $x_{23}+x_{14}+x_{12}\mod 2$\\
313 & 1111110010100111 & 3 & $(\bar{x}_1,\bar{x}_2,\bar{x}_3,\bar{x}_2)$ & 3 & $x_{23}+x_{14}+x_{134}+x_{12}+x_{124}+x_{123}\mod 3$\\
314 & 1111110010101000 & 2 & $(0,\bar{x}_1,\bar{x}_2,0)$ & 2 & $x_{23}+x_{14}\mod 3$\\
315 & 1111110010101001 & 2 & $(0,\bar{x}_1,\bar{x}_2,0)$ & 2 & $x_{23}+x_{14}\mod 2$\\
316 & 1111110011000000 & 2 & $(0,\bar{x}_1,\bar{x}_2,0)$ & 2 & $x_{23}+x_{13}+x_{12}\mod 2$\\
317 & 1111110011000010 & 2 & $(0,\bar{x}_1,\bar{x}_2,0)$ & 3 & $x_{23}+x_{13}+x_{12}+x_{124}\mod 3$\\
318 & 1111110011000011 & 2 & $(0,\bar{x}_1,\bar{x}_2,0)$ & 2 & $x_{23}+x_{13}+x_{12}\mod 3$\\
319 & 1111111000000000 & 3 & $(0,\bar{x}_1,\bar{x}_2,\bar{x}_3)$ & 3 & $x_{234}+x_{1}\mod 3$\\
320 & 1111111000000001 & 3 & $(0,\bar{x}_1,\bar{x}_2,\bar{x}_3)$ & 3 & $x_{234}+x_{1}\mod 2$\\
321 & 1111111000010000 & 3 & $(0,\bar{x}_1,\bar{x}_2,\bar{x}_3)$ & 3 & $x_{234}+x_{1}+x_{134}\mod 2$\\
322 & 1111111000010001 & 3 & $(0,\bar{x}_1,\bar{x}_2,\bar{x}_3)$ & 3 & $x_{234}+x_{1}+2x_{134}+x_{124}+x_{123}\mod 3$\\
323 & 1111111000010100 & 3 & $(0,\bar{x}_1,\bar{x}_2,\bar{x}_3)$ & 3 & $x_{234}+x_{1}+2x_{134}+2x_{124}+x_{123}\mod 3$\\
324 & 1111111000010101 & 3 & $(0,\bar{x}_1,\bar{x}_2,\bar{x}_3)$ & 3 & $x_{234}+x_{1}+x_{134}+x_{124}\mod 2$\\
325 & 1111111000010110 & 3 & $(0,\bar{x}_1,\bar{x}_2,\bar{x}_3)$ & 3 & $x_{234}+x_{1}+x_{134}+x_{124}+x_{123}\mod 2$\\
326 & 1111111000010111 & 3 & $(0,\bar{x}_1,\bar{x}_2,\bar{x}_3)$ & 3 & $x_{234}+x_{1}+x_{13}+x_{134}+x_{12}+x_{124}\mod 3$\\
327 & 1111111001000000 & 3 & $(0,\bar{x}_1,\bar{x}_2,\bar{x}_3)$ & 3 & $x_{234}+x_{1}+x_{14}+x_{134}+x_{124}\mod 2$\\
328 & 1111111001000001 & 3 & $(0,\bar{x}_1,\bar{x}_2,\bar{x}_3)$ & 3 & $x_{234}+x_{1}+2x_{14}+x_{134}+x_{124}\mod 3$\\
329 & 1111111001000010 & 3 & $(0,\bar{x}_1,\bar{x}_2,\bar{x}_3)$ & 3 & $x_{234}+x_{1}+2x_{14}+x_{134}+x_{124}+2x_{123}\mod 3$\\
330 & 1111111001000011 & 3 & $(0,\bar{x}_1,\bar{x}_2,\bar{x}_3)$ & 3 & $x_{234}+x_{1}+x_{14}+x_{134}+x_{124}+x_{123}\mod 2$\\
331 & 1111111001010000 & 3 & $(0,\bar{x}_1,\bar{x}_2,\bar{x}_3)$ & 3 & $x_{234}+x_{1}+2x_{14}+x_{124}\mod 3$\\
332 & 1111111001010001 & 3 & $(0,\bar{x}_1,\bar{x}_2,\bar{x}_3)$ & 3 & $x_{234}+x_{1}+x_{14}+x_{124}\mod 2$\\
333 & 1111111001010010 & 3 & $(0,\bar{x}_1,\bar{x}_2,\bar{x}_3)$ & 3 & $x_{234}+x_{1}+x_{14}+x_{124}+x_{123}\mod 2$\\
334 & 1111111001010011 & 3 & $(0,\bar{x}_1,\bar{x}_2,\bar{x}_3)$ & 3 & $x_{234}+x_{1}+2x_{14}+x_{12}+x_{123}\mod 3$\\
335 & 1111111001010100 & 3 & $(0,\bar{x}_1,\bar{x}_2,\bar{x}_3)$ & 3 & $x_{234}+x_{1}+x_{14}\mod 2$\\
336 & 1111111001010101 & 3 & $(0,\bar{x}_1,\bar{x}_2,\bar{x}_3)$ & 3 & $x_{234}+x_{1}+2x_{14}+x_{12}+2x_{124}+2x_{123}\mod 3$\\
337 & 1111111001010110 & 3 & $(0,\bar{x}_1,\bar{x}_2,\bar{x}_3)$ & 3 & $x_{234}+x_{1}+2x_{14}+x_{12}+2x_{124}+x_{123}\mod 3$\\
338 & 1111111001010111 & 3 & $(0,\bar{x}_1,\bar{x}_2,\bar{x}_3)$ & 3 & $x_{234}+x_{1}+x_{14}+x_{123}\mod 2$\\
339 & 1111111001100000 & 3 & $(0,\bar{x}_1,\bar{x}_2,\bar{x}_3)$ & 3 & $x_{234}+x_{1}+2x_{14}+2x_{13}+x_{124}+x_{123}\mod 3$\\
340 & 1111111001100001 & 3 & $(0,\bar{x}_1,\bar{x}_2,\bar{x}_3)$ & 3 & $x_{234}+x_{1}+x_{14}+x_{13}+x_{124}+x_{123}\mod 2$\\
341 & 1111111001100100 & 3 & $(0,\bar{x}_1,\bar{x}_2,\bar{x}_3)$ & 3 & $x_{234}+x_{1}+x_{14}+x_{13}+x_{123}\mod 2$\\
342 & 1111111001100101 & 3 & $(0,\bar{x}_1,\bar{x}_2,\bar{x}_3)$ & 3 & $x_{234}+x_{1}+2x_{14}+2x_{13}+x_{12}+2x_{124}\mod 3$\\
343 & 1111111001100110 & 3 & $(0,\bar{x}_1,\bar{x}_2,\bar{x}_3)$ & 3 & $x_{234}+x_{1}+2x_{14}+2x_{13}+x_{12}+2x_{124}+2x_{123}\mod 3$\\
344 & 1111111001100111 & 3 & $(0,\bar{x}_1,\bar{x}_2,\bar{x}_3)$ & 3 & $x_{234}+x_{1}+x_{14}+x_{13}\mod 2$\\
345 & 1111111001101000 & 3 & $(0,\bar{x}_1,\bar{x}_2,\bar{x}_3)$ & 3 & $x_{234}+x_{1}+x_{14}+x_{13}+x_{12}\mod 2$\\
346 & 1111111001101001 & 3 & $(0,\bar{x}_1,\bar{x}_2,\bar{x}_3)$ & 3 & $x_{234}+x_{1}+2x_{14}+2x_{13}+2x_{12}+2x_{124}+2x_{123}\mod 3$\\
347 & 1111111001110000 & 3 & $(0,\bar{x}_1,\bar{x}_2,\bar{x}_3)$ & 3 & $x_{234}+x_{1}+x_{14}+x_{13}+x_{134}+x_{124}+x_{123}\mod 2$\\
348 & 1111111001110100 & 3 & $(0,\bar{x}_1,\bar{x}_2,\bar{x}_3)$ & 3 & $x_{234}+x_{1}+2x_{14}+2x_{13}+x_{134}+x_{123}\mod 3$\\
349 & 1111111001110110 & 3 & $(0,\bar{x}_1,\bar{x}_2,\bar{x}_3)$ & 3 & $x_{234}+x_{1}+x_{14}+x_{13}+x_{134}\mod 2$\\
350 & 1111111001111000 & 3 & $(0,\bar{x}_1,\bar{x}_2,\bar{x}_3)$ & 3 & $x_{234}+x_{1}+2x_{14}+2x_{13}+x_{134}+2x_{12}+2x_{123}\mod 3$\\
351 & 1111111001111001 & 3 & $(0,\bar{x}_1,\bar{x}_2,\bar{x}_3)$ & 3 & $x_{234}+x_{1}+x_{14}+x_{13}+x_{134}+x_{12}\mod 2$\\
352 & 1111111001111100 & 3 & $(0,\bar{x}_1,\bar{x}_2,\bar{x}_3)$ & 3 & $x_{234}+x_{1}+x_{14}+x_{13}+x_{134}+x_{12}+x_{124}\mod 2$\\
353 & 1111111001111101 & 3 & $(0,\bar{x}_1,\bar{x}_2,\bar{x}_3)$ & 3 & $x_{234}+x_{1}+2x_{14}+2x_{13}+x_{134}+2x_{12}+x_{124}+2x_{123}\mod 3$\\
354 & 1111111001111110 & 3 & $(0,\bar{x}_1,\bar{x}_2,\bar{x}_3)$ & 3 & $x_{234}+x_{1}+2x_{14}+2x_{13}+x_{134}+2x_{12}+x_{124}+x_{123}\mod 3$\\
355 & 1111111001111111 & 3 & $(0,\bar{x}_1,\bar{x}_2,\bar{x}_3)$ & 3 & $x_{234}+x_{1}+x_{14}+x_{13}+x_{134}+x_{12}+x_{124}+x_{123}\mod 2$\\
356 & 1111111010000000 & 3 & $(0,\bar{x}_1,\bar{x}_2,\bar{x}_3)$ & 3 & $x_{234}+x_{14}+x_{13}+x_{134}+x_{12}+x_{124}+x_{123}\mod 2$\\
357 & 1111111010000001 & 3 & $(0,\bar{x}_1,\bar{x}_2,\bar{x}_3)$ & 3 & $x_{234}+x_{14}+x_{13}+x_{12}+2x_{123}\mod 3$\\
358 & 1111111010010000 & 3 & $(0,\bar{x}_1,\bar{x}_2,\bar{x}_3)$ & 3 & $x_{234}+x_{14}+x_{13}+x_{134}+x_{12}\mod 3$\\
359 & 1111111010010001 & 3 & $(0,\bar{x}_1,\bar{x}_2,\bar{x}_3)$ & 3 & $x_{234}+x_{14}+x_{13}+x_{12}+x_{124}+x_{123}\mod 2$\\
360 & 1111111010010100 & 3 & $(0,\bar{x}_1,\bar{x}_2,\bar{x}_3)$ & 3 & $x_{234}+x_{14}+x_{13}+x_{12}+x_{123}\mod 2$\\
361 & 1111111010010101 & 3 & $(0,\bar{x}_1,\bar{x}_2,\bar{x}_3)$ & 3 & $x_{234}+x_{14}+x_{13}+x_{134}+x_{12}+x_{124}\mod 3$\\
362 & 1111111010010110 & 3 & $(0,\bar{x}_1,\bar{x}_2,\bar{x}_3)$ & 3 & $x_{234}+x_{14}+x_{13}+x_{134}+x_{12}+x_{124}+x_{123}\mod 3$\\
363 & 1111111010010111 & 3 & $(0,\bar{x}_1,\bar{x}_2,\bar{x}_3)$ & 3 & $x_{234}+x_{14}+x_{13}+x_{12}\mod 2$\\
364 & 1111111011000000 & 3 & $(0,\bar{x}_1,\bar{x}_2,\bar{x}_3)$ & 3 & $x_{234}+x_{13}+x_{12}+2x_{123}\mod 3$\\
365 & 1111111011000001 & 3 & $(0,\bar{x}_1,\bar{x}_2,\bar{x}_3)$ & 3 & $x_{234}+x_{13}+x_{12}+x_{123}\mod 2$\\
366 & 1111111011000010 & 3 & $(0,\bar{x}_1,\bar{x}_2,\bar{x}_3)$ & 3 & $x_{234}+x_{13}+x_{12}\mod 2$\\
367 & 1111111011000011 & 3 & $(0,\bar{x}_1,\bar{x}_2,\bar{x}_3)$ & 3 & $x_{234}+x_{13}+2x_{12}+2x_{124}\mod 3$\\
368 & 1111111011010000 & 3 & $(0,\bar{x}_1,\bar{x}_2,\bar{x}_3)$ & 3 & $x_{234}+x_{13}+x_{134}+x_{12}+x_{123}\mod 2$\\
369 & 1111111011010010 & 3 & $(0,\bar{x}_1,\bar{x}_2,\bar{x}_3)$ & 3 & $x_{234}+x_{13}+2x_{134}+x_{12}+x_{124}+x_{123}\mod 3$\\
370 & 1111111011010011 & 3 & $(0,\bar{x}_1,\bar{x}_2,\bar{x}_3)$ & 3 & $x_{234}+x_{13}+x_{134}+x_{12}\mod 2$\\
371 & 1111111011100000 & 3 & $(0,\bar{x}_1,\bar{x}_2,\bar{x}_3)$ & 3 & $x_{234}+x_{134}+x_{12}\mod 2$\\
372 & 1111111011100001 & 3 & $(0,\bar{x}_1,\bar{x}_2,\bar{x}_3)$ & 3 & $x_{234}+x_{134}+x_{12}\mod 3$\\
373 & 1111111011100100 & 3 & $(0,\bar{x}_1,\bar{x}_2,\bar{x}_3)$ & 3 & $x_{234}+x_{134}+x_{12}+2x_{124}\mod 3$\\
374 & 1111111011100101 & 3 & $(0,\bar{x}_1,\bar{x}_2,\bar{x}_3)$ & 3 & $x_{234}+x_{134}+x_{12}+x_{124}\mod 2$\\
375 & 1111111011100110 & 3 & $(0,\bar{x}_1,\bar{x}_2,\bar{x}_3)$ & 3 & $x_{234}+x_{134}+x_{12}+x_{124}+x_{123}\mod 2$\\
376 & 1111111011100111 & 3 & $(0,\bar{x}_1,\bar{x}_2,\bar{x}_3)$ & 3 & $x_{234}+x_{134}+2x_{12}+x_{124}+x_{123}\mod 3$\\
377 & 1111111011101000 & 3 & $(0,\bar{x}_1,\bar{x}_2,\bar{x}_3)$ & 3 & $x_{234}+x_{134}+x_{124}+x_{123}\mod 3$\\
378 & 1111111011101001 & 3 & $(0,\bar{x}_1,\bar{x}_2,\bar{x}_3)$ & 3 & $x_{234}+x_{134}+x_{124}+x_{123}\mod 2$\\
379 & 1111111100000000 & 1 & $(\bar{x}_1,0,0,0)$ & 1 & $x_{1}\mod 2$\\
380 & 1111111110000000 & 2 & $(\bar{x}_1,0,0,\bar{x}_2)$ & 2 & $x_{14}+x_{13}+x_{12}\mod 5$\\
381 & 1111111110000001 & 2 & $(\bar{x}_1,0,0,\bar{x}_2)$ & 2 & $x_{14}+x_{13}+x_{12}\mod 3$\\
382 & 1111111110010000 & 2 & $(\bar{x}_1,0,0,\bar{x}_2)$ & 2 & $x_{14}+4x_{13}+2x_{12}\mod 5$\\
383 & 1111111110010100 & 2 & $(\bar{x}_1,0,0,\bar{x}_2)$ & 2 & $x_{14}+2x_{13}+2x_{12}\mod 3$\\
384 & 1111111110010110 & 2 & $(\bar{x}_1,0,0,\bar{x}_2)$ & 2 & $3x_{14}+3x_{13}+3x_{12}\mod 6$\\
385 & 1111111111000000 & 2 & $(\bar{x}_1,0,\bar{x}_2,0)$ & 2 & $x_{13}+x_{12}\mod 3$\\
386 & 1111111111000010 & 3 & $(\bar{x}_1,\bar{x}_2,\bar{x}_2,\bar{x}_3)$ & 3 & $x_{13}+x_{12}+x_{124}+x_{123}\mod 3$\\
387 & 1111111111000011 & 2 & $(\bar{x}_1,0,\bar{x}_2,0)$ & 2 & $x_{13}+x_{12}\mod 2$\\
388 & 1111111111100000 & 3 & $(\bar{x}_1,0,\bar{x}_2,\bar{x}_3)$ & 3 & $x_{134}+x_{12}\mod 3$\\
389 & 1111111111100001 & 3 & $(\bar{x}_1,0,\bar{x}_2,\bar{x}_3)$ & 3 & $x_{134}+x_{12}\mod 2$\\
390 & 1111111111100100 & 3 & $(\bar{x}_1,0,\bar{x}_2,\bar{x}_3)$ & 3 & $x_{134}+x_{12}+x_{124}\mod 2$\\
391 & 1111111111100110 & 3 & $(\bar{x}_1,0,\bar{x}_2,\bar{x}_3)$ & 3 & $x_{134}+2x_{12}+x_{124}+x_{123}\mod 3$\\
392 & 1111111111100111 & 3 & $(\bar{x}_1,0,\bar{x}_2,\bar{x}_3)$ & 3 & $x_{134}+x_{12}+x_{124}+x_{123}\mod 2$\\
393 & 1111111111101000 & 3 & $(\bar{x}_1,0,\bar{x}_2,\bar{x}_3)$ & 3 & $x_{134}+x_{124}+x_{123}\mod 2$\\
394 & 1111111111101001 & 3 & $(\bar{x}_1,0,\bar{x}_2,\bar{x}_3)$ & 3 & $x_{134}+x_{124}+x_{123}\mod 3$\\
395 & 1111111111110000 & 2 & $(\bar{x}_1,\bar{x}_2,0,0)$ & 2 & $x_{12}\mod 2$\\
396 & 1111111111111000 & 3 & $(\bar{x}_1,\bar{x}_2,0,\bar{x}_3)$ & 3 & $x_{124}+x_{123}\mod 3$\\
397 & 1111111111111001 & 3 & $(\bar{x}_1,\bar{x}_2,0,\bar{x}_3)$ & 3 & $x_{124}+x_{123}\mod 2$\\
398 & 1111111111111100 & 3 & $(\bar{x}_1,\bar{x}_2,\bar{x}_3,0)$ & 3 & $x_{123}\mod 2$\\
399 & 1111111111111110 & 4 & $(\bar{x}_1,\bar{x}_2,\bar{x}_3,\bar{x}_4)$ & 4 & $x_{1234}\mod 2$

\end{longtable}

\end{document}